\documentclass[a4paper,12pt]{article}

\usepackage{graphicx}
\usepackage[english]{babel}
\usepackage{amsfonts}
\usepackage{amsmath}
\usepackage{amssymb}
\usepackage{amsthm}
\usepackage[pdftex,colorlinks=true,linkcolor=blue,citecolor=blue,urlcolor=red,unicode=true,hyperfootnotes=true,bookmarksnumbered]{hyperref}
\usepackage[margin=1truein]{geometry}
\usepackage{enumitem}
\usepackage{dsfont}
\usepackage{multirow}

\theoremstyle{plain}
\newtheorem{theorem}{Theorem}
\newtheorem{lemma}[theorem]{Lemma}
\newtheorem{corollary}[theorem]{Corollary}
\newtheorem{proposition}[theorem]{Proposition}
\newtheorem{claim}[theorem]{Claim}

\theoremstyle{definition}
\newtheorem{remark}[theorem]{Remark}
\newtheorem{definition}[theorem]{Definition}

\begin{document}

\title{First order complexity of finite random structures}
\author{Danila Demin\footnote{Department of Discrete Mathematics, Moscow Institute of Physics and Technology, Dolgoprudny 141700, Russia; Department of Mathematics and Computer Science, St Petersburg University, Saint Petersburg 199034, Russia}, \quad Maksim Zhukovskii\footnote{Department of Computer Science, University of Sheffield, Sheffield S1 4DP, UK}}
\date{}
\maketitle

\begin{abstract}

For a sequence of random structures with $n$-element domains over a relational signature, we define its first order (FO) complexity as a certain subset in the Banach space $\ell^{\infty}/c_0$.  The well-known FO zero-one law and FO convergence law correspond to FO complexities equal to $\{0,1\}$ and a subset of $\mathbb{R}$, respectively. We present a hierarchy of FO complexity classes, introduce a stochastic FO reduction that allows to transfer complexity results between different random structures, and deduce using this tool several new logical limit laws for binomial random structures.   Finally, we introduce a conditional distribution on graphs, subject to a FO sentence $\varphi$, that generalises certain well-known random graph models, show instances of this distribution for every complexity class, and prove that the set of all $\varphi$ validating 0--1 law is not recursively enumerable.

\end{abstract}

\section{Introduction}
\label{sc:intro}

Let $\sigma$ be a relational signature. A {\it random $n$-structure} over the signature $\sigma$ is a random element of the set of all structures with domain $[n]:=\{1,\ldots,n\}$ and over $\sigma$. Commonly it is supported by a set of $\sigma$-axioms $\mathcal{F}$, i.e. it is assumed that with probability $1$ the random $n$-structure models $\mathcal{F}$. Let  $\sigma=\{P_1,\ldots,P_s\}$, where $P_i$ has arity $d_i$ (we adopt a usual convention that $\sigma$ includes the equality relation whose axioms are part of the logic). We denote the random $n$-structure uniformly distributed over all $n$-structures over $\sigma$ by $D^{\sigma}(n\mid\mathcal{F})$ or $D^{(d_1,\ldots,d_s)}(n\mid\mathcal{F})$.

For example, the well known binomial random graph $G(n,p)$ is a random $n$-structure over $\sigma=\{=,\sim\}$, where $=$ represents the coincidence of vertices, and $\sim$ represents the graph adjacency relation. The axioms in 
$$
\mathcal{F}=
\{\forall x\forall y\,\,(x\sim y)\Leftrightarrow(y\sim x),\,\,
\forall x\,\, \neg(x\sim x)\}
$$
allow to define the distribution on the set of undirected graphs without loops. 
 The distribution of $G(n,p)$ is defined as follows: for a fixed $n$-structure $G$ that models $\mathcal{F}$ (that is, $G$ is an undirected graph without loops), its probability equals $p^{|E(G)|}(1-p)^{{n\choose 2}-|E(G)|}$, where $E(G)$ is the set of unordered pairs $\{x,y\}$ such that $x\sim y$, (that is, the set of edges of $G$). In other words, edges appear independently with probability $p$. In particular, $G(n,1/2)=D^{\sigma}(n\mid\mathcal{F})=D^{(2)}(n\mid\mathcal{F})$.

Another example is a binomial random $n$-structure 
$$
D^{\sigma}(n,p_1,\ldots,p_s)=D^{(d_1,\ldots,d_s)}(n,p_1,\ldots,p_s)
$$ 
over  $\sigma=\{=,P_1,\ldots,P_s\}$ that does not have any predefined axioms: each interpretation appears with probability $p_1^{N_1}\ldots p_s^{N_s}(1-p_1)^{n^{d_1}-N_1}\ldots(1-p_s)^{n^{d_s}-N_s}$, where $N_i$ is the number of $d_i$-tuples satisfying $P_i$. In other words, $D^{(d_1,\ldots,d_s)}(n,p_1, \dots, p_s)$ is the binomial random hypergraph with hyperedges (ordered multisets) of cardinalities $d_1,\ldots,d_s$, where each hyperedge of ``type $P_i$'' appears with probability $p_i$ independently of all the others. In particular, for a signature consisting of one predicate $P$ of arity $d$, we get a binomial random directed $d$-uniform hypergraph
 $D^{(d)}(n, p)$. 

The following fundamental result in finite model theory, known as the first order (FO) 0--1 law, was proven by Glebskii, Kogan, Liogonkii, Talanov~\cite{GKLT} and independently by Fagin~\cite{Fa}: for a fixed finite relational signature $\sigma$, any FO sentence $\varphi$ over $\sigma$ is either true on asymptotically almost all $n$-structures or false. In other words, $\mathrm{Pr}(D^{(d_1,\ldots,d_s)}(n, 1/2, \ldots, 1/2)\models\varphi)$ converges either to $0$ or to $1$ as $n\to\infty$. Same arguments can be used to show that $D^{(d_1,\ldots,d_s)}(n,p_1,\ldots,p_s)$ obeys the FO 0--1 law for all constant $p_1,\ldots,p_s\in(0,1)$. Fagin \cite{Fa} also proved that the same is true for graphs, i.e. $G(n, p)$ obeys the FO 0--1 law. 

Since then, the validity of the FO 0--1 law was studied for many other random structures: binomial random graphs with $p=p(n)$~\cite{Ly92,SS,Sp}, random regular graphs~\cite{HK}, random trees~\cite{McColm}, recursive random graphs~\cite{Kleinbergs,MalyshZhuk}, random geometric graphs~\cite{McColm_circle}, random graphs embeddable on a surface~\cite{HMNT}, random structures $D^{\sigma}(n\mid\mathcal{F})$ where $\mathcal{F}$ is a set of parametric axioms~\cite{EF_finite,Oberschelp},  and many others~\cite{SpBook,SZ,Winkler}. However, no methods have been developed to transfer logical limit laws between different random structures. In particular, Fagin applied the same proof as for $D^{(d_1,\ldots,d_s)}(n,p_1,\ldots,p_s)$ to $G(n,p)$ instead of transferring the law. One of the main contributions of our paper is such a transferring tool. In particular, it allows to transfer the FO 0--1 law from $D^{(2)}(n,p)$ to the binomial random graph as well as to the binomial random {\it directed} graph without loops.

Before moving on to a more detailed discussion of our results, let us notice an important application of logical limit laws to the study of hierarchy of logics and time complexity. Indeed, if a random structure satisfies the 0--1 (or convergence) law for a language $\mathcal{L}_1$ while it does not satisfy the law for $\mathcal{L}_2\supseteq\mathcal{L}_1$, then the inclusion $\mathcal{L}_2\supset\mathcal{L}_1$ is strict. This simple observation was used in~\cite{VZ1,VZ2} to show the lower bound on the minimum quantifier depth and the minimum number of variables of a FO sentence describing the property of containing an induced subgraph isomorphic to a fixed given graph $F$. The latter fact implies certain limitations of the respective method of solving the induced subgraph isomorphism decision problem: the validity of a FO sentence with $k$ variables on an $n$-vertex graph is decidable in time $O(n^k)$, see details in~\cite{VZ1,VZ2}. \\


For $G(n,p)$ with $p=p(n)=o(1)$, a breakthrough achievement was obtained by Shelah and Spencer. First of all note that the above mentioned classical FO 0--1 law for $G(n,p)$ can be generalised to all $p=p(n)$ such that $\min\{p,1-p\}n^{\alpha}\to\infty$ as $n\to\infty$ for all $\alpha>0$ (see~\cite{Sp}). So it is natural to further consider $p=n^{-\alpha}$, $\alpha>0$. Shelah and Spencer \cite{SS} proved the following:
\begin{itemize}
\item If $\alpha\in(0,1]$ is rational, then $G(n,n^{-\alpha})$ does not obey the FO 0--1 law.
\item If $\alpha\in(0,1]$ is irrational, then $G(n,n^{-\alpha})$ obeys the FO 0--1 law.
\item If $\alpha=1+\frac{1}{k}$ for some $k\in\mathbb{N}$, then $G(n,n^{-\alpha})$ does not obey the FO 0--1 law.
\item If either $p=o(n^{-2})$, or for some $k\in\mathbb{N}$, $n^{-1-1/k}\ll p\ll n^{-1-1/(k+1)}$, then $G(n,p)$ obeys the FO 0--1 law.
\end{itemize}
Our tool can be also used to show that the FO 0--1 law does not hold: for example, for $\alpha>1$, the failure of the FO 0--1 law for $G(n,n^{-\alpha})$ transfers to the failure of the FO 0--1 law for $D^{(d+1)}(n,(d+\alpha-2)\frac{\ln n}{n})$ for every integer $d\geq 2$.

\begin{remark}
The random graph $G(n,n^{-\alpha})$ for $\alpha>1$ is very sparse: with asymptotical probability 1 ({\it with high probability} or, for brevity, {\it whp} in what follows) it is a forest consisting only of tree components of bounded sizes~(see, e.g.,~\cite{Janson}). In contrast, $D^{(d+1)}(n,\Theta(\ln n/n))$ is weakly connected whp~\cite{Po}. Since the validity of the FO 0--1 law for $G(n,n^{-\alpha})$ in this case follows immediately from standard properties of the logical equivalence (see, e.g., \cite[Section 4.2]{OZ}), so it is not surprising that our tool does not transfer the validity of FO 0--1 laws from  $G(n,n^{-\alpha})$ to $D^{(d+1)}(n,\Theta(\ln n/n))$.
\end{remark}

A special interest in combinatorial and probabilistic community was chained to properties of specifically $G(n,c/n)$ for constant $c$ because of so-called phase transition phenomenon~\cite{ER-evol} --- in particular, the emergence of a giant component. Lynch \cite{Ly92} proved that $G(n,c/n)$ obeys the {\it FO convergence law} (i.e. for every FO sentence the probability of its truth on $G(n,c/n)$ converges to some number in $[0, 1]$ as $n\to\infty$), however the FO 0--1 law does not hold. Note that Shelah and Spencer \cite{SS} disproved even the FO convergence law for $G(n,n^{-\alpha})$ when $\alpha\in(0,1)$ is rational. Larrauri, M\"{u}ller, and Noy~\cite{LMN} described possible limits of truth probabilities of FO sentences on $G(n,c/n)$. They proved that the closure of the set of limits in $[0, 1]$ consists of a finite number of segments and determined the minimum constant $c_0$ for which this segment is unique and coincides with $[0, 1]$. Also, they generalised this result to $d$-uniform unoriented hypergraphs. 

We transfer the upper bound for this threshold $c_0$ from $d$-uniform unoriented hypergraphs to $d$-uniform $H$-hypergraphs, where $H$ is an arbitrary subgroup of the symmetry group $S_d$ (i.e. a hyperedge is an orbit under the action of $H$ on $d$-tuples $(x_1, \dots, x_d)$ of distinct elements from a fixed $d$-element set). For example, unoriented hypergraphs correspond to the case $H=S_d$. Regarding the lower bound, which coincides with the upper bound, it can be transferred from oriented hypergraphs (i.e. $H=\{\mathrm{id}\}$) to $H$-hypergraphs for any $H$. Luckily, the same proof method as from \cite{LMN} can be applied to prove the tight lower bound for oriented hypergraphs as well. Thus, the problem of finding the threshold $c_0$ for any way of assigning an orientation to hyperedges can be reduced to two extreme cases $H=S_d$ and $H=\{\mathrm{id}\}$. 

Another application of our tool is transferring the FO 0--1 laws from binomial random structures to random structures $D^{\sigma}(n\mid\mathcal{F})$ where $\mathcal{F}$ is a set of parametric axioms with an additional requirement that an empty structure satisfies the axioms. For such structures, the 0-1 law was proved by Oberschelp~\cite{Oberschelp} applying, again, a similar strategy as for the usual unconditional binomial model.\\

Our tool is a certain preorder on sequences of random $n$-struc\-tu\-res, $n\in\mathbb{N}$, which we call {\it the stochastic FO reduction}. This preorder expresses a hierarchy of sequences of random structures: for higher sequences in this preorder, their FO limit behaviour is more complex. In particular, for each pair of sequences $A$ and $B$ such that $A$ is reducible to $B$ and $B$ obeys the FO 0--1 law (convergence law), $A$ obeys the FO 0--1 law (convergence law) as well. Besides the  above mentioned applications of the stochastic FO reduction to transferring logical laws, we prove that the stochastic FO reduction preorder defines {\it stable} equivalence classes of $G(n,p)$: a little change of $p$ does not change the equivalence class. For example, this observation allows to transfer the absence of the FO 0--1 law from $D^{(1)}(n, \frac{c}{n})$ to $D^{(d+1)}(n, \frac{d \ln n}{n})$. The stochastic FO reduction is defined in Section~\ref{SC:SFOR}.

The notion of stochastic FO reduction as well as different asymptotic logical behaviour of different random structures naturally lead to a concept of {\it FO complexity} of a sequence of random structures $D_n$. The FO complexity has to generically describe the limit behaviour of $\mathrm{Pr}(D_n\models\varphi)$ over all FO sentences $\varphi$ in such a way that if $A$ stochastically reduces to $B$, then $B$ is at least as complex~as~$A$. 

The FO complexity is defined in Section~\ref{FOC}. Formally, we define it as $\mathcal{D}/c_0$, where $\mathcal{D}\subset\ell^{\infty}$ is exactly the set of all sequences $(\mathrm{Pr}(D_n\models\varphi))_{n\in\mathbb{N}}$, and $c_0$ is the set of sequences converging to $0$. In particular, if a sequence of random structures obeys the FO 0--1 law, then its FO complexity is $\{0,1\}$ (for brevity, we identify a constant sequence with its element), and, if it obeys the FO convergence law, then its FO complexity is exactly the set of all limits of probabilities of FO sentences. Note that the above mentioned result of Larrauri, M\"{u}ller, and Noy \cite{LMN} guarantees that the closure of the FO complexity of $G(n,c/n)$ consists of finitely many segments. However, it is not hard to see that there are binomial random structures such that their FO complexities are even not totally bounded. In particular, this is the case for the binomial random graph $G(n,n^{-\alpha})$ for rational $\alpha<1$ as we show in Section~\ref{binFOC}. It is also possible to define $p(n)$ in a way such that the FO complexity of $G(n,p(n))$ spans an infinite-dimensional subspace as well, but the complexity is totally bounded. 

Finally, we consider the random graph 
$$
G(n\mid\varphi):=D^2(n\mid\varphi\wedge\text{symmetric}\wedge\text{anti-reflexive})
$$ 
chosen uniformly at random from the set of all undirected graphs without loops that satisfy a given FO axiom $\varphi$. This model generalises the well-studied binomial random graph $G(n,1/2)$, random regular graphs~\cite{HK,Wo}, random union of disjoint cliques~\cite{GLMNP}, and random permutations~\cite{Ko}. For every complexity class, we show the existence of the respective axiom $\varphi$: the FO complexity of $G(n\mid\varphi)$ may be trivial (i.e. the FO 0--1 law holds), may span a 1-dimensional subspace of $\ell^{\infty}/c_0$ (the FO convergence law holds), may span a $k$-dimensional subspace for every positive integer $k$, may be totally bounded but span an infinite-dimensional subspace, and may not be totally bounded. These examples are given in Sections~\ref{axFOC},~\ref{sc:proofiii}.

Note that the FO 0--1 law holds for $G(n\mid\varphi)$ whenever $\varphi$ is true on $G(n,1/2)$ with probability bounded away from 0. The latter may only happen when probability $\mathrm{Pr}(G(n,1/2)\models\varphi)$ approaches 1. Since the FO almost sure theory of $G(n,1/2)$ can be axiomatised by {\it extension axioms} (see,~e.g.,~\cite{SpBook}) and probability that $G(n,1/2)$ does not satisfy the $k$-th extension axiom is at most $n^k(1-2^{-k})^{n-k}$, it is easy to see that the problem of determining whether $\mathrm{Pr}(G(n,1/2)\models\varphi)$ approaches 1 is decidable. Nevertheless, we show in Section~\ref{sc:rec_enum} that the problem of determining, for an arbitrary input $\varphi$, whether $G(n\mid\varphi)$ obeys the FO 0--1 law  is even not recursively enumerable. It worths noting that related questions were studied in the literature. In~\cite{Liogonkii_conditional}, Liogon'kii proved that, given FO sentences $\varphi$ and $\psi$ over $\sigma$, it is undecidable whether $\lim_{n\to\infty}\mathrm{Pr}(D^{\sigma}(n\mid\varphi)\models\psi) $ exists. Several extensions of this result were proved in \cite{GHK_conditional}: in particular, computing or approximating the limit, given the fact that it exists, is still undecidable.\\ 

A preliminary short version of this paper was presented at the 39th Annual ACM/IEEE Symposium on Logic in Computer Science, LICS 2024~\cite{LICS_version}.

\section{FO complexity}
\label{FOC}

In Section~\ref{defFOC}, we define {\it the FO complexity} of a sequence of random structures that generalises the FO 0--1 law and the FO convergence law. After that, we show a strict hierarchy of FO complexity classes of binomial random graphs (Section~\ref{binFOC}) and conditional random graphs subject to FO sentences (Section~\ref{axFOC}). The most essential part of the main theorem in Section~\ref{axFOC} that asserts the existence of a FO sentence $\varphi$ such that the complexity of $G(n\mid\varphi)$ is totally bounded and infinite dimensional is proved in Section~\ref{sc:proofiii}. This proof develops a method of constructing FO sentences that define isomorphism classes of certain asymmetric graphs and may be of its own interest.

\subsection{Definition of complexity and hierarchy of random structures}\label{defFOC}

Let us recall necessary definitions. The Banach space $\ell^{\infty}$ is the linear space of all bounded sequences of real numbers $x=(x_1, x_2, \dots)$ with the norm $\|x\|=\sup _{n \in \mathbb{N}}|x_n|$. The Banach space $c_0$ is a subspace of $\ell^{\infty}$, which consists of all vectors $x$ such that $\lim_{n \to \infty} |x_n| = 0$. The Banach space $c$ is a subspace of $\ell^{\infty}$, which consists of all vectors $x$ such that $\lim_{n \to \infty} |x_n|$ exists. Norms on $c_0$ and $c$ are induced from $\ell^{\infty}$. The Banach space $\ell^{\infty}/c_0$ is a quotient space, which consists of classes $x+c_0$ with the norm $\limsup_{n \to \infty}|x_n|$. The canonical projection $\pi: \ell^{\infty} \to \ell^{\infty}/c_0$ maps each $x \in \ell^{\infty}$ to $x+c_0 \in \ell^{\infty}/c_0$. We denote by $X/c_0$ the image of a subset $X \subset \ell^{\infty}$ under~$\pi$.

\begin{definition}
Let $D_n$, $n\in\mathbb{N}$, be a sequence of random $n$-struc\-tu\-res. The {\it FO complexity $\mathrm{FOC}(D_n)$} of $D_n$ is the set $\mathcal{D}/c_0$, where $\mathcal{D}\subset\ell^{\infty}$ is the set of all sequences $(\mathrm{Pr}(D_n\models\varphi))_{n\in\mathbb{N}}$ over all FO sentences $\varphi$.
\end{definition}

Due to the next straightforward proposition (we give a proof for the sake of completeness), the validity of the FO 0--1 law or the validity of the FO convergence law are the cases of the smallest FO complexities. For brevity, for any real $\lambda$, we denote vectors $(\lambda, \lambda, \dots) \in \ell^{\infty}$ and $\pi(\lambda, \lambda, \dots) \in \ell^{\infty}/c_0$ by $\lambda$. Note that $c/c_0\cong\mathbb{R}$, and there exists an isomorphism that maps $\lambda\in c/c_0$ to $\lambda\in\mathbb{R}$. Therefore, we identify any subset of $c/c_0$ with the set of respective numbers in $\mathbb{R}$.

\begin{claim}\label{lawsFOC}
Let $D_n$, $n\in\mathbb{N}$, be a sequence of random structures.
\begin{enumerate}
\item[(i)] $D_n$ satisfies the FO 0--1 law iff $\mathrm{FOC}(D_n)=\{0,1\}$.
\item[(ii)] $D_n$ satisfies the FO convergence law iff $\mathrm{FOC}(D_n) \subset \mathbb{R}$.
\item[(iii)] The set of limits of sequences $(\mathrm{Pr}(D_n\models\varphi))_{n\in\mathbb{N}}$ coincides with $\mathrm{FOC}(D_n) \cap \mathbb{R}$.
\end{enumerate}
\end{claim}

\begin{proof}
Let $x \in \ell^{\infty}$. From the definition, $\pi(x)-\lambda=0$ is equivalent to $\lim_{n \to \infty}|x_n - \lambda| = 0$, i.e. $\lim_{n \to \infty}x_n = \lambda$. Since there is an isomorphism between $c/c_0$ and $\mathbb{R}$ that maps any $\pi(x)$, $x\in c$, to $\lim_{n\to\infty} x_n$, we have (iii). Therefore, $\mathrm{FOC}(D_n) \subset c/c_0$ is equivalent to the fact that each sequence $(\mathrm{Pr}(D_n\models\varphi))_{n\in\mathbb{N}}$ converges, i.e. we have (ii). Similarly, we get (i).
\end{proof}

Examples of random structures with these smallest FO complexities are well known: (1) $\mathrm{FOC}(D_n)=\{0,1\}$ for the uniformly random structure $D_n=D^{\sigma}(n,1/2,\ldots,1/2)$, as well as for the uniformly chosen random graph $D_n=G(n,1/2)$ \cite{GKLT, Fa}; (2) $\mathrm{FOC}(D_n)=\{0,1\}$ for the binomial random graph $D_n = G(n, p)$ with $pn^\alpha\to\infty$ for all $\alpha > 0$ \cite{Sp} or $p=n^{-\alpha}$, where $\alpha$ is either irrational or bigger than $1$ and not equal $1+1/m$ for any positive integer $m$ \cite{SS}; (3) $\mathrm{FOC}(D_n) \subset \mathbb{R}$ for the binomial random graph $D_n = G(n, n^{-\alpha})$ for $\alpha = 1 + \frac{1}{m}$ \cite{SS}; (4) $\mathrm{FOC}(D_n) \subset \mathbb{R}$ for the binomial random graph $D_n = G(n, \frac{c}{n})$ for $c > 0$ \cite{Ly92}, and the closure of $\mathrm{FOC}(D_n)$ consists of finitely many segments \cite{LMN}. In the next section, we present random structures with $d$-dimensional, infinite dimensional but totally bounded, as well as not totally bounded FO complexities. 

\subsection{Complexity of $G(n,p)$}\label{binFOC}

We first show that it is possible to construct very sparse binomial random graphs (consisting of only isolated vertices and isolated edges whp) $G(n,p)$ with all the properties of FO complexities mentioned in Section~\ref{sc:intro}. However, the respective sequences $p=p(n)$ are quite artificial and far from being `smooth'. So, later in this section we show that the FO complexity of $G(n,n^{-\alpha})$ for rational $\alpha\in(0,1)$ is not totally bounded, and that all the properties are achievable by $G(n\mid\varphi)$ for appropriately chosen FO sentences $\varphi$. We shall use the following technical claim that follows from the fact that whp $G(n,p=o(n^{-3/2}))$ consists of isolated vertices and isolated edges, and the number of isolated edges can be approximated by Poisson random variables $\mathrm{Pois}(\lambda_n = p{n\choose 2})$ (see,~e.g.,~\cite{Janson,SpBook}).

\begin{claim}\label{CL:VECTORS}
Let $p = o(n^{-3/2})$, $\lambda_n = p(n) {n\choose 2}$, and, for every $k \in\mathbb{Z}_{\geq 0}$, $x_k = \pi \left((\frac{\lambda_n^k}{k!}e^{-\lambda_n})_{n \in \mathbb{N}}\right)$. Then $\mathrm{FOC}(G(n, p))$ is a union of the set $X$ of all finite sums of vectors $x_k$ and the set $1-X:=\{1-x,\,x\in X\}$.
\end{claim}

\begin{proof}
Let a FO sentence $\varphi_{\geq k}$ express the property of being a disjoint union of at least $k$ edges. The sentence $\varphi_k = \varphi_{\geq k} \wedge \neg \varphi_{\geq k+1}$ expresses the property of being a disjoint union of exactly $k$ edges. Each FO sentence $\psi \wedge \varphi_{\geq 0}$ is tautologically equivalent either to $\varphi_{i_1} \vee \cdots \vee \varphi_{i_s}$ or to $\varphi_{i_1} \vee \cdots \vee \varphi_{i_s} \vee \varphi_{\geq i_{s+1}}$ for some non-negative integers $i_1 < \cdots < i_s < i_{s+1}$. Since $p = o(n^{-\frac{3}{2}})$, the sequence of probabilities $(\mathrm{Pr}(G(n, p)\models\varphi_{\geq 0}))_{n\in\mathbb{N}}$ converges to one (see~\cite[Theorem 3.6.2]{SpBook}). Hence, $\mathrm{Pr}(G(n, p)\models\psi) - \mathrm{Pr}(G(n, p)\models \psi \wedge \varphi_{\geq 0})$ converges to zero. 

Let us show that for $\varphi_i$, the equality $\pi(\mathrm{Pr}(G(n, p)\models\varphi_{i})) = x_i$ holds. It is easier to compute the probability that the graph has exactly $i$ edges without any restriction on overlapping edges. Luckily, these two probabilities are asymptotically equal. Let a FO sentence $\hat{\varphi}_{i}$ express the property of containing exactly $i$ edges (not necessary disjoint). So, the FO sentence $\hat{\varphi}_{i} \wedge \varphi_{\geq 0}$ is tautologically equivalent to $\varphi_i$. Since $\varphi_{\geq 0}$ is true whp, we get
$$
\pi(\mathrm{Pr}(G(n, p)\models\varphi_{i}))  = \pi(\mathrm{Pr}(G(n, p)\models\hat{\varphi_{i}} \wedge \varphi_{\geq 0})) = \pi(\mathrm{Pr}(G(n, p)\models\hat{\varphi}_{i})).
$$
Now, it is sufficient to prove that $\pi(\mathrm{Pr}(G(n, p)\models\hat{\varphi}_{i})) = x_i$.
\begin{align*}
\mathrm{Pr}(G(n, p)\models\hat{\varphi}_{i})  & ={{n\choose 2}\choose i} p^i (1 - p)^{{n\choose 2}-i} \\
&  = (1+o(1)) \frac{1}{i!} {n\choose 2}^i p^i e^{{n\choose 2} \ln (1 - p)} (1 - p)^{-i} = \frac{\lambda_n^i}{i!} e^{-\lambda_n} + o(1)
\end{align*}
as needed.

Since $\varphi_i$ contradicts $\varphi_j$ for each $j \not = i$, $\pi(\mathrm{Pr}(G(n, p)\models\varphi_{i_1} \vee \cdots \vee \varphi_{i_s})) = x_{i_1} + \cdots + x_{i_s}$. Therefore, we get $X \subset \mathrm{FOC}(G(n, p(n)))$. For the sentence $\psi = \varphi_{i_1} \vee \cdots \vee \varphi_{i_s} \vee \varphi_{\geq i_{s+1}}$, we consider a sentence $\neg\psi$, for which $\pi(\mathrm{Pr}(G(n, p)\models\neg\psi)) = 1 - \pi(\mathrm{Pr}(G(n, p)\models\psi))$. Note that the sentence $\neg \psi \wedge \varphi_{\geq 0}$ is equivalent to $\bigvee_{j < i_{s+1}, j \not = i_1, \dots, i_s} \varphi_j$. So, $\pi(\mathrm{Pr}(G(n, p)\models\psi)) \in 1 - X$, and we have $\mathrm{FOC}(G(n, p)) = X \cup (1 - X)$, completing the proof.
\end{proof}


\begin{theorem}\label{thBinFOC}
For $p=p(n)$, we let $G_n\sim G(n,p)$.
\begin{enumerate}
\item[(i)] For each $d \geq 1$, there is a sequence $p(n) \in [0, 1]$ such that $\mathrm{FOC}(G_n)$ spans a $d$-dimensional subspace of $\ell^{\infty}/c_0$.
\item[(ii)] There is a sequence $p(n)\in [0, 1]$ such that $\mathrm{FOC}(G_n)$ is totally bounded but spans an infinite-dimensional subspace of $\ell^{\infty}/c_0$.
\item[(iii)] There is a sequence $p(n) \in [0, 1]$ such that $\mathrm{FOC}(G_n)$ is not totally bounded.
\end{enumerate}
\end{theorem}

\begin{proof}
To prove (i), consider $p =\lambda_n/{n\choose 2}$, where $\lambda_n$ equals $n$ modulo $d$. The sequence $(\frac{\lambda_n^k}{k!}e^{-\lambda_n})_{n \in \mathbb{N}}$ is $d$-periodic, i.e. its $(n+d)$-th element equals the $n$-th element. The subspace $L_d$ of $d$-periodic sequences in $\ell^{\infty}$ is $d$-dimensional. Consider the following basis in $L_d$: let $e_r$ be the vector with ones on $(d t + r)$-th positions, $t\in\mathbb{Z}_+$, and zeros on all others. Let us prove that the system of vectors $f_k = (\frac{\lambda_n^k}{k!}e^{-\lambda_n})_{n \in \mathbb{N}}$ for $0 \leq k \leq d-1$ is also a basis in $L_d$. The vector $f_k$ equals $\sum\limits_{r=0}^{d-1} \frac{r^k}{k!}e^{-r} e_r$, and
\begin{align*}
\mathrm{det}\left(\left(\frac{i^j}{j!}e^{-i}\right)_{0\leq i,j\leq d-1}\right) =\mathrm{det}\left((i^j)_{0\leq i,j\leq d-1}\right) \prod_i e^{-i} \prod_j \frac{1}{j!}
 = \prod_{i < i'} (i' - i) \prod_i e^{-i} \prod_j \frac{1}{j!} \not = 0 .
\end{align*}
Then, $f_k$, $0\leq k\leq d-1$, is a basis in $L_d$. Since $x_k = \pi (f_k)$ for all $k$, we have that $\langle x_k,\, k\geq 0\rangle$ coincides with the subspace $\pi(L_d)$ in $\ell^{\infty}/c_0$. The intersection of $L_d$ and $c_0$ is trivial because each $d$-periodic vector in $c_0$ is zero. Therefore, the space $\pi(L_d)$ is also $d$-dimensional. Note that $\langle X \rangle=\pi(L_d)$ and $1\in\pi(L_d)$ implying that $\langle 1-X \rangle=\pi(L_d)$. Since $f_k\in X$ for all $0\leq k\leq d-1$, by Claim~\ref{CL:VECTORS}, we get that $\langle \mathrm{FOC}(G(n,p)) \rangle =\pi(L_d)$ completing the proof. \\

To prove (ii), consider $p = \frac{2 \lambda_n}{n(n-1)}$, where $\lambda_n$ equals $\frac{1}{r}$ if $n$ is divisible by $2^r$, but not divisible by $2^{r+1}$. We will denote $r(n)$ the maximum $r$ such that $n$ is divisible by $2^r$. Thus, $\lambda_n=\frac{1}{r(n)}$. Let $e_r$ be the vector with ones on $2^r(2t+1)$-th positions, and zeros on all others. The vector $(\frac{\lambda_n^k}{k!}e^{-\lambda_n})_{n \in \mathbb{N}}$ equals $\sum\limits_{r=0}^{\infty} \frac{1}{r^k k!}e^{-\frac{1}{r}} e_r$. Vectors $\pi(e_r)$ are nonzero and linearly independent. Therefore, $x_k = \sum\limits_{r=0}^{\infty} \frac{1}{r^k k!}e^{-\frac{1}{r}} \pi(e_r)$. Similarly, as for (i), $x_k$ are linearly independent. So, the set $X$ spans an infinite-dimensional subspace. 

Let us show that the sum $\sum\limits_{k=0}^{\infty} \|x_k\|$ converges. 
$$
\|x_k\| = \limsup_n \frac{1}{(r(n))^k k!}e^{-\frac{1}{r(n)}} \leq \limsup_n \frac{1}{(r(n))^{k}k!} = \frac{1}{k!} .
$$
Since the sum $\sum\limits_{k=0}^{\infty} \frac{1}{k!}$ converges, we get the convergence of the considered sum. For each positive $\varepsilon$, we can choose an integer $N$ such that $\sum\limits_{k=N+1}^{\infty} \|x_k\| < \varepsilon$. Then, for each vector $v = x_{i_1}+x_{i_2}+\cdots+x_{i_s} \in X$ consider a vector $v'= x_{i_1} + x_{i_2} + \cdots + x_{i_{s'}}$, where $i_{s'}$ is the greatest number among $i_j$ such that $i_{j} \leq N$. Since $v - v'$ is expressible as a sum of vectors $x_k$ with $k > N$, we get $\|v - v'\| \leq \varepsilon$. Therefore, the set $X$ can be covered by finitely many balls of radius $\varepsilon$ because there is finitely many sums of vectors $x_k$ with $k \leq N$. Similarly for $1-X$. Hence, $X \cup (1-X)$ is totally bounded. \\

To prove (iii), consider $p = 2 \lambda_n/(n(n-1))$, where $\lambda_n$ equals $m(r(n))$, and $m(r)$, $r\in\mathbb{Z}_+$, is defined in such a way that, for some $k(r)$, vectors $y_r=x_{0}+\cdots+x_{k(r)}$ are at distances at least $\frac{1}{3}$ from each other. Let us construct such $m(r)$ and $k(r)$. For every non-negative integer $k$, consider a function $g_k:\mathbb{R}\to\mathbb{R}$ defined as follows: $g_k(\lambda)=(1 + \lambda + \cdots + \lambda^k/k!)e^{-\lambda}$. This sequence of functions satisfies two properties: (a) $\lim\limits_{\lambda \to \infty} g_k(\lambda) = 0$ for every fixed $k$;  (b) $\lim\limits_{k \to \infty} g_k(\lambda) = 1$ for every fixed $\lambda$. Let $m(0) = 0$ and $k(0) = 0$. For each $r \geq 1$, we define $m(r)$ as the least integer $M$ such that $\forall m \geq M: g_{k(r-1)}(m) \leq \frac{1}{3}$. Such an integer $M$ exists due to the property (a). Next, we define $k(r)$ as the least integer $K$ such that $\forall k \geq K: g_{k}(m(r)) \geq \frac{2}{3}$. Such an integer $K$ exists due to the property (b).  The vector $y_r$ equals $\pi \left((g_{k(r)}(m(r(n))))_{n \in \mathbb{N}}\right)$. So, $y_r$ and $y_s$, $r \geq s+1$, are at distance
\begin{align*}
\|y_r - y_s\| & = \limsup_n \biggl |g_{k(r)}(m(r(n))) - g_{k(s)}(m(r(n))) \biggr |\\
& \geq g_{k(r)}(m(s+1)) - g_{k(s)}(m(s+1)) \geq \frac{2}{3} - \frac{1}{3} = \frac{1}{3}.
\end{align*}
Since there is a sequence of vectors $y_r \in X$ such that $\|y_r - y_s\| \geq \frac{1}{3}$ for each pair of distinct positive integers $r$ and $s$, the set $X \cup (1 - X)$ is not totally bounded.
\end{proof}

Theorem~\ref{thSSFOC} stated below claims that $G(n,n^{-\alpha})$ has a not totally bounded complexity for every rational $\alpha\in(0,1)$. We derive it using a construction of a FO sentence introduced by Shelah and Spencer in~\cite{SS} to disprove the convergence law. 
\begin{lemma}[Shelah, Spencer~\cite{SS}]\label{absCLSS}
Let $G_n\sim G(n, n^{-\alpha})$. For every integer $d\geq 100$, there exists a FO sentence $\varphi_d$ such that
\begin{itemize}
\item[(i)] if $\log^* n \equiv \lfloor \frac{d}{4} \rfloor (\bmod\; d)$ then $G_n\models\varphi_d$ whp;
\item[(ii)] if $\log^* n \equiv \lfloor \frac{3d}{4} \rfloor (\bmod\; d)$ then $G_n\models\neg\varphi_d$ whp,
\end{itemize}
where $\log^*$ denotes the iterated logarithm, i.e. the number of times the logarithm need to be applied to the number to make it one or less.
\end{lemma}

In the original paper \cite{SS} Lemma~\ref{absCLSS} was formulated only for  $d=100$, but literally the same proof works for any $d \geq 100$ (the lower bound 100 could be sufficiently improved, though it is not important for us).

\begin{theorem}\label{thSSFOC}
For any rational $\alpha \in (0, 1)$, $\mathrm{FOC}(G(n, n^{-\alpha}))$ is not totally bounded. 
\end{theorem}

\begin{proof}
For every prime number $p\geq 100$, let us consider the sentence $\varphi_p$ whose existence is claimed by Lemma~\ref{absCLSS}. Let $x_p(n)$ be the probability $\mathrm{Pr}(G(n, n^{-\alpha}) \models \varphi_p)$. Then, vectors $x_p = \pi \left( (x_p(n))_{n \in \mathbb{N}} \right)$ are in $\mathrm{FOC}(G(n, n^{-\alpha}))$. Let $(y_p(n))_{n \in \mathbb{N}} \in \ell^{\infty}$ be the sequence such that 
\begin{itemize}
\item[(i)] if $\log^* n \equiv \lfloor \frac{p}{4} \rfloor (\bmod\; p)$ then $y_p(n)=1$;
\item[(ii)] if $\log^* n \equiv \lfloor \frac{3p}{4} \rfloor (\bmod\; p)$ then $y_p(n)=0$;
\item[(iii)] $y_p(n) = x_p(n)$, otherwise.
\end{itemize}
The sequence $x_p(n) - y_p(n)$ converges to zero because, for each $n$ such that $\log^* n \equiv \lfloor \frac{p}{4} \rfloor (\bmod\; p)$ or $\log^* n \equiv \lfloor \frac{3p}{4} \rfloor (\bmod\; p)$, we have the required convergence due to the properties of $\varphi_p$ given by Lemma~\ref{absCLSS}, and $x_p(n) - y_p(n) = 0$ for all other $n$.

Thus, vectors $x_p$ and $\pi \left( (y_p(n))_{n \in \mathbb{N}} \right)$ are equal. By the Chinese remainder theorem, for each pair of distinct primes $p$ and $q$ there exists an integer number $m$ such that $m \equiv \lfloor \frac{p}{4} \rfloor (\bmod\; p)$ and $m \equiv \lfloor \frac{3q}{4} \rfloor (\bmod\; q)$. Therefore, for each $n$ such that $\log^* n \equiv m (\bmod\; pq)$, we have $y_p(n) - y_q(n) = 1$. Since there are infinitely many $n$ such that $\log^* n \equiv m (\mathrm{mod}\; pq)$, we have $\|x_p - x_q\| \geq 1$. Hence, we have infinitely many vectors in the set $\mathrm{FOC}(G(n, n^{-\alpha}))$, which are at distances at least $1$ from each other, i.e. $\mathrm{FOC}(G(n, n^{-\alpha}))$ is not totally bounded.
\end{proof}

So, indeed, $\mathrm{FOC}(G(n, n^{-\alpha}))$ is not totally bounded when $\alpha \in (0,1)\cap\mathbb{Q}$. We are not able to present a ``nice'' $p$ so that $\mathrm{FOC}(G(n, p))$ is either $d$-dimensional, $d>1$, or infinite-dimensional and totally bounded. However, this appears to be possible for $G(n\mid\varphi)$. 

\subsection{Complexity of $G(n\mid\varphi)$}\label{axFOC}

First of all, note that for any FO sentence $\varphi$, if $\liminf\limits_{n\to\infty} \mathrm{Pr}(G(n, \frac{1}{2}) \models \varphi) > 0$ (it actually may only happen when the limit is 1), then $\mathrm{FOC}(G(n\mid\varphi)) = \{ 0, 1 \}$ due to the FO 0--1 law for $G(n, \frac{1}{2})$. Moreover, there is a FO sentence $\varphi$ such that $G(n\mid\varphi)$ obeys the FO convergence law but not the FO 0--1 law. For example, consider $\varphi$ which expresses the property of consisting of isolated vertices and exactly one connected component of size $3$. Then, probability of containing a triangle converges to $\frac{1}{4}$.

\begin{theorem}\label{TH:AXFOC}
For a FO sentence $\varphi$, we let $G_n\sim G(n\mid\varphi)$.
\begin{enumerate}
\item[(i)] There is a FO sentence $\varphi$ such that $\mathrm{FOC}(G_n)$ is a dense subset of $[0, 1]$.
\item[(ii)] For each $d \geq 1$, there is a FO sentence $\varphi$ such that $\mathrm{FOC}(G_n)$ spans a $d$-dimensional subspace of $\ell^{\infty}/c_0$.
\item[(iii)] There is a FO sentence $\varphi$ such that $\mathrm{FOC}(G_n)$ is totally bounded but spans an infinite-dimensional subspace of $\ell^{\infty}/c_0$.
\item[(iv)] There is a FO sentence $\varphi$ such that $\mathrm{FOC}(G_n)$ is not totally bounded.
\end{enumerate}
\end{theorem}

We postpone the proof of (iii) to Section~\ref{sc:proofiii}: it is long enough to interrupt the flow of the paper and it requires an additional background that we outline in the beginning of Section~\ref{sc:proofiii}.

\begin{proof}[Proof of parts (i), (ii) and (iv) of Theorem~\ref{TH:AXFOC}]
{\bf To prove (i),} consider a FO sentence $\varphi$ which expresses the property of being $2$-regular. For the respective random graph $G(n\mid\varphi)$, the FO convergence law was proven by Lynch in~\cite{Ly05}. Then, $\mathrm{FOC}(G(n\mid\varphi))$ is a subset of $[0, 1]$. To prove that this subset is dense, we refer to the result proven by Bollob\'{a}s and Wormald~\cite{Bo, Wo80, Wo81}.

\begin{lemma}[Bollob\'{a}s, Wormald \cite{Bo, Wo80, Wo81}]\label{lmWo}
Fix an integer $d\geq 2$ and $C>0$. In random uniform $d$-regular graphs on $[n]$, the vectors of numbers of cycles of length $\ell \leq C$ converge in distribution to a vector of independent Poisson random variables $\mathrm{Pois}\left((d-1)^\ell/(2\ell)\right)$.
\end{lemma}

Let a FO sentence $\psi_{\ell}$ express the property of containing a cycle of length $\ell$. For a $\{0, 1\}$-word $W$ of length $w$, let $\psi_W$ be a conjunction of sentences $\psi_\ell$ if $W(\ell-2)=1$, and $\neg \psi_\ell$ if $W(\ell-2)=0$. For each pair of distinct words $W$ and $W'$ of length $w$, $\psi_W$ contradicts $\psi_{W'}$. Also, the disjunction of $\psi_W$ over all words $W$ of length $w$ is a tautology. Therefore, for an arbitrary labelling $W_1, W_2, \dots, W_{2^w}$ of all such words, we have that 
$$
q_s := \lim\limits_{n \to \infty} \mathrm{Pr}\left(G(n\mid\varphi) \models \bigvee_{i = 1}^s \psi_{W_i}\right) = \sum_{i = 1}^s \lim\limits_{n \to \infty} \mathrm{Pr}(G(n\mid\varphi) \models \psi_{W_i}).
$$
Note that $q_0 = 0$ and $q_{2^w} = 1$. Also, from Lemma~\ref{lmWo}, we have that
\begin{align*}
q_s - q_{s+1}  & = \lim\limits_{n \to \infty} \mathrm{Pr}(G(n\mid\varphi) \models \psi_{W_s}) \\
&=\prod_{W_s(\ell-2)=1}(1 - e^{-\frac{1}{2\ell}}) \prod_{W_s(\ell-2)=0} e^{-\frac{1}{2\ell}} \leq \prod_{\ell = 3}^{w+2} e^{-\frac{1}{2\ell}} < \frac{e^{\frac{3}{2}}}{\sqrt{w+2}}.
\end{align*}
Therefore, $q_s$ is an increasing sequence of numbers in $[0, 1]$ such that $q_0 = 0$, $q_{2^w} = 1$ and $q_s - q_{s-1} < e^{\frac{3}{2}}/\sqrt{w+2}$. Hence, for each number $x \in [0, 1]$, there is an element of this sequence such that $|x - q_s| < \frac{1}{2} e^{\frac{3}{2}}/\sqrt{w+2}$. Since $q_s$ are limiting probabilities for FO sentences, and $w$ can be chosen arbitrary large, $\mathrm{FOC}(G(n\mid\varphi))$ is a dense subset of $[0, 1]$ as needed. 

{\bf To prove (ii),} consider a FO sentence $\varphi$ which expresses the property of being a disjoint union of $d$-cliques and at most one $r$-clique for some $0 \leq r < d$. Note that, for each $n$, this property defines a single isomorphism class. 
 Fix a FO sentence $\psi$. Since for any graphs $A$, $B$, there exists $m_0 \in \mathbb{N}$ such that, for any $m \geq m_0$, graphs $m_0 A \sqcup B$ and $m A \sqcup B$ are not distinguishable by $\psi$ (see~\cite{Co}) we get that there exists $t_0$ such that for all $t \geq t_0$ graphs on $[dt_0+r]$ and $[dt+r]$ that satisfy $\varphi$ are not distinguishable by $\psi$. Then, the sequence $(\mathrm{Pr}(G(n\mid\varphi) \models \psi))_{n \in \mathbb{N}}$ consists only of zeros and ones and is $d$-periodic for $n$ large enough. Therefore, $\mathrm{FOC}(G(n\mid\varphi))$ is a set of projections of $d$-periodic sequences of ones and zeros, i.e. is contained in the $d$-dimensional subspace $\pi(L_d) \subset \ell^{\infty}/c_0$, where $L_d$ is the $d$-dimensional space of all $d$-periodic sequences in~$\ell^{\infty}$.

For each $0 \leq r < d$ consider a FO sentence $\psi_r$ which expresses the property of containing an isolated $r$-clique. The $n$-th element of the sequence $\mathrm{Pr}(G(n\mid\varphi) \models \psi_r)$ equals one if $n \equiv r (\mathrm{mod}\; d)$, and zero otherwise. Projections of $(\mathrm{Pr}(G(n\mid\varphi) \models \psi_r)_{n \in \mathbb{N}}$, $0 \leq r < d$, generate the space $\pi(L_d)$. Thus, indeed $\mathrm{FOC}(G(n\mid\varphi))$ spans a $d$-dimensional subspace of $\ell^{\infty}/c_0$. 

{\bf To prove (iv),} we need the following definition.

\begin{definition}\label{defCartesian}
Let $G$ and $H$ be two graphs. {\it The Cartesian product} $G \Box H$ is the graph on $V(G) \times V(H)$ with adjacency relation $(u, v) \sim (u', v') \Leftrightarrow ((u \sim u') \wedge (v = v')) \vee ((u = u') \wedge (v \sim v'))$. In other words, every edge of the Cartesian product belongs either to an induced subgraph $G_v \cong G$ on $\{(u, v) \mid u \in V(G)\}$ for some $v\in V(H)$ or to an induced subgraph $H_u \cong H$ on $\{(u, v) \mid v \in V(H)\}$ for some $u\in V(G)$.
\end{definition}

Let us present a FO sentence $\varphi$ that describes the property of being isomorphic to $K_s \Box K_t$, for some $s, t > 0$. We construct this sentence as a conjunction of three sentences describing the following properties:

\begin{enumerate}
\item[(a)] For each vertex $v$, its neighbourhood consists of two disjoint cliques $A_v$ and $B_v$;
\item[(b)] For every pair of non-adjacent vertices $u$ and $v$, there is a unique edge from $u$ to $A_v$ and a unique edge from $u$ to $B_v$;
\item[(c)] For every vertex $v$ and any two its non-adjacent neighbours $x\neq y$, there is a unique vertex $u \not = v$ adjacent to $x$ and $y$. 
\end{enumerate}

Let us prove that this sentence expresses the desired property. Observe that $K_s \Box K_t$ satisfies it. Let us then consider any graph $G$ satisfying this sentence. Consider a vertex $v$, from (a) we have two cliques $A_v$ and $B_v$. For any vertex $x$ in $A_v$, we have two cliques $A_x = A_v \cup \{v\} \backslash \{x\}$ and $B_x$. The clique $B_x$ does not intersect $A_x$ by the property (a), and does not intersect $B_v$ because there is no edges between $x$ and $B_v$. Also, we conclude that each vertex adjacent to $x$ and not adjacent to $v$ is a vertex of $B_x$. Similarly, for $y \in B_v$, $B_y = B_v \cup \{v\} \backslash \{y\}$, $A_y$ does not intersect both $B_y$ and $A_v$, and each vertex adjacent to $y$ and not adjacent to $v$ is a vertex of $A_y$. 

From (b) and (c), we have that vertices $u$ which are not adjacent to $v$ are in a one-to-one correspondence with pairs of vertices $x \in A_v$ and $y \in B_v$: each such $u$ is adjacent to the respective $x$ and $y$, and it is not adjacent to any other vertex in $A_v \cup B_v$. Therefore, for each pair $x \in A_v$ and $y \in B_v$, cliques $B_x$ and $A_y$ have a unique common vertex  $u$. This implies that all cliques $B_x$ have the same size as $B_v$, and all cliques $A_y$ have the same size as $A_v$. Hence, the graph $G$ consists of cliques $\{v\} \cup A_v$ and $\{y\} \cup A_y$, for all $y \in B_v$, and cliques $\{v\} \cup B_v$ and $\{x\} \cup B_x$, for all $x \in A_v$, i.e. isomorphic to $K_s \Box K_t$, for some $s, t > 0$.

Now, we fix this FO sentence $\varphi$ describing the property of being isomorphic to $K_s \Box K_t$, for some $s, t > 0$. Let, for $n\in\mathbb{Z}_{>0}$, $D(n)$ be the set of all divisors of $n$. Let $\mu_d=2$, if $d = \sqrt{n}$, and $\mu_d=1$, otherwise. For each $d \in D(n)$ there are $\mu_d d!\left(\frac{n}{d}\right)!$ automorphisms of the graph $K_d \Box K_{\frac{n}{d}}$. Therefore, there are $\frac{n!}{\mu_d d!\left(\frac{n}{d}\right)!}$ graphs on $[n]$ isomorphic to $K_d \Box K_{\frac{n}{d}}$. The probability that $G(n\mid\varphi)$ is isomorphic to $K_d \Box K_{\frac{n}{d}}$ equals
$$
\frac{n!/\mu_d}{ d!\left(\frac{n}{d}\right)!}/\left(\sum_{d' \in D(n)} \frac{n!}{2 d'!\left(\frac{n}{d'}\right)!}\right) = \frac{1/\mu_d}{ d!\left(\frac{n}{d}\right)!}/\left(\sum_{d' \in D(n)} \frac{1}{2 d'!\left(\frac{n}{d'}\right)!}\right)
$$
where we have $2$ in the denominator of the normalisation factor instead of $\mu_{d'}$ since we count twice every graph $K_{d'}\Box K_{n/d'}$ when $d'\neq\sqrt{n}$. 

Let $\psi_d$ be a FO sentence which expresses the property of containing an inclusion-maximal clique of size $d$, i.e. a clique of size $d$ which is not included in any clique of size $d+1$. If a graph on $[n]$ satisfies $\varphi$ then it is isomorphic to $K_d \Box K_{\frac{n}{d}}$ if and only if it satisfies $\psi_d$. Consider $d_1 < d_2$, let us prove that, for each $\varepsilon > 0$ there are infinitely many numbers $n$ such that $\mathrm{Pr}(G(n\mid\varphi) \models \psi_{d_1})-\mathrm{Pr}(G(n\mid\varphi) \models \psi_{d_2}) > 1-\varepsilon$. Let $n = d_1 p$, where $p$ is a prime number bigger than $d_2$. Therefore, $n$ is not divisible by $d_2$, and $\mathrm{Pr}(G(n\mid\varphi) \models \psi_{d_2}) = 0$. Since $d_1 < d_2 < p$, each divisor $m$ of $n$ such that $m \leq \sqrt{n}$ cannot be divisible by $p > \sqrt{n}$. Hence, such divisors are divisors of $d_1$. Thus, we can estimate the normalisation factor for the probability $\mathrm{Pr}(G(n\mid\varphi) \models \psi_{d_1})$ in the following way:
\begin{align*}
\sum_{d' \in D(n)} \frac{1/2}{d'!\left(\frac{n}{d'}\right)!} &= \sum_{d' \in D(d_1)} \frac{1}{d'!\left(\frac{n}{d'}\right)!} \leq \frac{1}{d_1!\left(\frac{n}{d_1}\right)!} + \frac{|D(d_1)| - 1}{\left(\frac{2n}{d_1}\right)!} \\
&\leq\frac{1}{d_1!\left(\frac{n}{d_1}\right)!}\left(1 + \frac{(d_1+1)!p!}{(2p)!}\right) \leq \frac{1}{d_1!\left(\frac{n}{d_1}\right)!}\left(1 + \frac{p!p!}{(2p)!}\right) \leq \frac{1+2^{-p}}{d_1!\left(\frac{n}{d_1}\right)!}.
\end{align*}
Therefore, for all primes $p$ such that $(1 + 2^{-p})^{-1} > 1 - \varepsilon$, we have $\mathrm{Pr}(G(n\mid\varphi) \models \psi_{d_1}) > 1 - \varepsilon$. So, each pair $\psi_{d_1}$ and $\psi_{d_2}$ defines a pair of vectors in $\mathrm{FOC}(G(n\mid\varphi))$ at distance at least $1$, and $\mathrm{FOC}(G(n\mid\varphi))$ is not totally bounded.
\end{proof}


\section{Stochastic FO reduction}
\label{SC:SFOR}

In this section we define a {\it stochastic FO reduction} and describe its useful properties (in Section~\ref{sc:SFOR_def}). Then we show its effectiveness by using it to derive certain logical limit laws for dense (in Section~\ref{sc:SFOR_dense}) and sparse (in Section~\ref{sc:SFOR_sparse}) relational structures as well as to transfer higher FO complexities between random relational structures. Finally, in Section~\ref{sc:SFOR_Muller} we use the stochastic FO reduction to generalise the result of Larrauri, M\"{u}ller, and Noy about the closure of FO complexity of binomial random $d$-hypergraphs with $p=c/n^{d-1}$ from undirected hypergraphs to directed hypergraphs for any possible way to choose an orientation of hyperedges.

\subsection{Definition and main properties}
\label{sc:SFOR_def}

Let $\sigma,\sigma'$ be two signatures; $\mathcal{D}_n$ and $\mathcal{D}'_n$ be the sets of all finite structures on $[n]$ over $\sigma$ and $\sigma'$ respectively; $\mathcal{D}=\sqcup_{n\in\mathbb{N}}\mathcal{D}_n$ and $\mathcal{D}'=\sqcup_{n\in\mathbb{N}}\mathcal{D}'_n$; $D_n,D'_n$ be random relational $n$-structures over $\sigma,\sigma'$ respectively. Moreover, for any two FO sentences $\varphi,\varphi'$ over $\sigma,\sigma'$ respectively, let $\mathcal{D}(\varphi)\subset\mathcal{D}$ and $\mathcal{D}'(\varphi')\subset\mathcal{D}'$ be sets of all structures satisfying $\varphi$ and $\varphi'$ respectively. Finally, let us consider algebras $\mathcal{A}=\{\mathcal{D}(\varphi)\}$, $\mathcal{A}'=\{\mathcal{D}'(\varphi')\}$ (recalling that an {\it algebra} is closed under finite unions in contrast to a $\sigma$-algebra). For an $(\mathcal{A}\mid\mathcal{A}')$-measurable function $f:\mathcal{D}\to\mathcal{D}'$ and a FO sentence $\varphi'$, we denote by $f^{-1}(\varphi')=:\varphi$ a FO sentence such that $\mathcal{D}(\varphi)=f^{-1}(\mathcal{D}'(\varphi'))$.

\begin{definition}
\textbf{A stochastic FO reduction} from $D'=(D'_n)_{n\in\mathbb{N}}$ to $D=(D_n)_{n\in\mathbb{N}}$ is an $(\mathcal{A}\mid\mathcal{A}')$-measurable function $f:\mathcal{D}\to\mathcal{D}'$ such that, for every $n\in\mathbb{N}$, $f$ maps $n$-structures to $n$-structures, and $\lim\limits_{n \to \infty}|\mathrm{Pr}(D_n \models f^{-1}(\varphi')) - \mathrm{Pr}(D_n' \models \varphi')| = 0$ for every FO sentence $\varphi'$ over $\sigma'$.
\end{definition}

If there is a stochastic FO reduction from $D'$ to $D$, we say that $D'$ is {\it reducible} to $D$ (or sometimes we say that $D'_n$ is reducible to $D_n$ meaning of course a reduction of the entire sequences) and denote it as $D' \preceq D$ (or $D_n' \preceq D_n$). Let us first observe that the basic property of being a preorder, that holds for reductions in the computational complexity theory, holds for our reduction as well.

\begin{claim}\label{redpreorder}
The stochastic FO reduction relation $\preceq$ is a preorder.
\end{claim}

\begin{proof}
Firstly, we prove that $\preceq$ is reflexive. Let we have a random relational structure $D_n$. We define a stochastic FO reduction by the identity mapping $id: \mathcal{D} \to \mathcal{D}$. It is clear that this mapping is $(\mathcal{A}\mid\mathcal{A})$-measurable, maps $n$-structures to $n$-structures, and satisfies $\lim\limits_{n \to \infty} |\mathrm{Pr}(D_n \models id^{-1}(\varphi)) - \mathrm{Pr}(D_n \models \varphi)| = 0$, implying the reflexivity.

Next, we prove the transitivity. Let we have two stochastic FO reductions $f: \mathcal{D} \to \mathcal{D}'$ from $D'_n$ to $D_n$ and $g: \mathcal{D}' \to \mathcal{D}''$ from $D''_n$ to $D'_n$. As a reduction from $D''_n$ to $D_n$ we use the composition $g \circ f$. A composition of $(\mathcal{A}\mid\mathcal{A}')$-measurable and $(\mathcal{A}'\mid\mathcal{A}'')$-measurable functions is $(\mathcal{A}\mid\mathcal{A}'')$-measurable. Also, a composition of two functions which map $n$-structures to $n$-structures, for each $n \in \mathbb{N}$, satisfies the same property. Finally,
\begin{align*}
0 & \leq \lim\limits_{n \to \infty}|\mathrm{Pr}(D_n \models (g \circ f)^{-1}(\varphi)) - \mathrm{Pr}(D_n'' \models \varphi)| \\ 
 &\leq \lim\limits_{n \to \infty}|\mathrm{Pr}(D_n \models  f^{-1}(g^{-1}(\varphi))) - \mathrm{Pr}(D_n' \models g^{-1}(\varphi))| + \\ 
 &\quad\quad\quad\quad\quad\quad\quad\quad\quad\quad+ \lim\limits_{n \to \infty}|\mathrm{Pr}(D_n' \models g^{-1}(\varphi)) - \mathrm{Pr}(D_n'' \models \varphi)| = 0.
\end{align*}
Hence, we have
$$
\lim\limits_{n \to \infty}|\mathrm{Pr}(D_n \models (g \circ f)^{-1}(\varphi)) - \mathrm{Pr}(D_n'' \models \varphi)| = 0,
$$ 
that finishes the proof.
\end{proof}

Next, we show key property of stochastic FO reductions which allow us to transfer FO complexities between different random relational structures.

\begin{claim}\label{redFOC}
Suppose $D'_n \preceq D_n$. Then, $\mathrm{FOC}(D_n') \subseteq \mathrm{FOC}(D_n)$.
\end{claim}
\begin{proof}
Let $f: \mathcal{D} \to \mathcal{D}'$ be a reduction from $D_n'$ to $D_n$. Consider a vector $v \in \mathrm{FOC}(D_n')$. There is a FO sentence $\varphi$ in the signature of $D_n'$ such that $v = (\mathrm{Pr}(D_n'\models\varphi))_{n\in\mathbb{N}} + c_0$. Since $f$ is the stochastic FO reduction, there is a FO sentence $f^{-1}(\varphi)$ in the signature of $D_n$ such that $\lim\limits_{n \to \infty}|\mathrm{Pr}(D_n \models f^{-1}(\varphi)) - \mathrm{Pr}(D_n' \models \varphi)| = 0$. Therefore, $v = (\mathrm{Pr}(D_n'\models\varphi))_{n\in\mathbb{N}} + c_0 = (\mathrm{Pr}(D_n\models f^{-1}(\varphi)))_{n\in\mathbb{N}} + c_0 \in \mathrm{FOC}(D_n)$.
\end{proof}

\begin{corollary}\label{corRedLaws}
Suppose $D'_n \preceq D_n$.
\begin{enumerate}
\item[(i)] If $D_n$ obeys the FO 0--1 law, then $D_n'$ obeys the FO 0--1 law as well.
\item[(ii)] If $D_n$ obeys the FO convergence law, then $D_n'$ obeys the FO convergence law as well.
\end{enumerate}
\end{corollary}

\begin{proof}
By Claim~\ref{lawsFOC} and Claim~\ref{redFOC}, we have that $\mathrm{FOC}(D_n') \subseteq \mathrm{FOC}(D_n) \subseteq \{0, 1\}$, for the case (i), and $\mathrm{FOC}(D_n') \subseteq \mathrm{FOC}(D_n) \subseteq \mathbb{R}$, for the case (ii). Then, by Claim~\ref{lawsFOC}, we obtain the assertion of the corollary.
\end{proof}

The opposite statement of Claim~\ref{redFOC} is false. Indeed, let $\sigma = \sigma' = \{U, V\}$, where $U$ and $V$ have arity $1$. We define distributions of $D_n$ and $D'_n$ in the following way:

\begin{itemize}
\item 
Let $P_n$ be the unique $n$-structure that satisfies $\forall x \,\, U(x) \wedge V(x)$. Set $\mathrm{Pr}(D'_n = P_n) = \mathrm{Pr}(D_n = P_n) = \frac{1}{4}$.

\item
Let $Q_n$ be the unique $n$-structure that satisfies $\forall x \,\, U(x) \wedge \neg V(x)$. Set $\mathrm{Pr}(D'_n = Q_n) = \frac{1}{4}$ and $\mathrm{Pr}(D_n = Q_n) = \frac{1}{2}$.

\item
Let $R_n$ be the unique $n$-structure that satisfies $\forall x \,\, \neg U(x) \wedge V(x)$. Set $\mathrm{Pr}(D'_n = R_n) = \frac{1}{4}$ and $\mathrm{Pr}(D_n = R_n) = 0$.

\item 
Let $S_n$ be the unique $n$-structure that satisfies $\forall x \,\, \neg U(x) \wedge \neg V(x)$. Set $\mathrm{Pr}(D'_n = S_n) = \mathrm{Pr}(D_n = S_n) = \frac{1}{4}$.
\end{itemize}

For each family of relational structures $\{P_n\}_{n\in\mathbb{N}}, \{Q_n\}_{n\in\mathbb{N}}, \{R_n\}_{n\in\mathbb{N}}, \{S_n\}_{n\in\mathbb{N}}$, we have that each two elements of one family are elementary equivalent (note that the equality is not included in signatures). Therefore, for each FO sentence $\varphi$ over $\sigma$, we have that the probabilities $\mathrm{Pr}(D_n \models \varphi)$ and $\mathrm{Pr}(D'_n \models \varphi)$ do not depend on $n$. Since almost surely $D'_n$ has four values, and $D_n$ has three values, complexities of these random structures equal sets of sums of elements of sub(multi)sets of the multisets $\{\!\!\{\frac{1}{4}, \frac{1}{4}, \frac{1}{4}, \frac{1}{4}\}\!\!\}$ and $\{\!\!\{\frac{1}{4}, \frac{1}{2}, \frac{1}{4}\}\!\!\}$ respectively. Hence, $\mathrm{FOC}(D_n) = \mathrm{FOC}(D_n') = \{ 0, \frac{1}{4}, \frac{1}{2}, \frac{3}{4}, 1 \}$.

Suppose, we have a stochastic FO reduction $f$ from $D'$ to $D$. It implies that $f^{-1}(P_n)$, $f^{-1}(Q_n)$, $f^{-1}(R_n)$, $f^{-1}(S_n)$ divide the set $\{P_n, Q_n, S_n\}$ into disjoint parts, and therefore one of them is empty. Without loss of generality $f^{-1}(P_n) \cap \{P_n, Q_n, S_n\}$ is empty. Then, $\mathrm{Pr}(D_n' \in f^{-1}(P_n)) = 0$ and $\lim\limits_{n \to \infty}|\mathrm{Pr}(D_n' \in f^{-1}(P_n)) - \mathrm{Pr}(D_n = P_n)| = \frac{1}{4} \not = 0$, contradiction.

\subsection{Application to FO zero-one laws for dense structures}
\label{sc:SFOR_dense}

Let us now use stochastic FO reductions to transfer FO 0--1 laws.

\subsubsection{Random graphs}

We consider the signature $\{=,\rightarrow\}$, where $\rightarrow$ has arity 2, and denote by $\vec{G}(n, p)$ the binomial random directed graph without loops on the set of vertices $[n]$, i.e. every directed edge (out of the set of $n(n-1)$ possible edges) appears independently of the others with probability $p$.

\begin{theorem}\label{redDigraphs}
Let $p\in(0,1)$ (not necessarily a constant). There are stochastic FO reductions $G(n, p^2) \preceq \vec{G}(n, p) \preceq D^\sigma(n, p)$, where $\sigma = \{=, P\}$. In particular, for a constant  $p \in (0, 1)$, $\vec{G}(n, p)$ obeys the FO 0--1 law.
\end{theorem}

\begin{proof}
Define the reduction from $\vec{G}(n, p)$ to $D^\sigma(n, p)$ as the function $f$ which deletes loops from a directed graph. This mapping can be ``defined" by the FO formula $\psi_{\rightarrow}(x, y) = P(x, y) \wedge \neg(x = y)$. So, for each FO sentence $\varphi$ over signature $\{=, \rightarrow\}$, the FO sentence $f^{-1}(\varphi)$ is constructed from $\varphi$ by replacing each $x \rightarrow y$ by $\psi_{\rightarrow}(x, y)$. Note that $f(D^{\sigma}(n,p))\stackrel{d}=\vec{G}(n,p)$. Thus, $\mathrm{Pr}(D^\sigma(n, p) \models f^{-1}(\varphi)) = \mathrm{Pr}(\vec{G}(n, p) \models \varphi)$ and $f$ is indeed a stochastic FO reduction.

From Corollary~\ref{corRedLaws} and the fact that $D^{\sigma}(n, p)$ obeys the FO 0--1 law, we have that $\vec{G}(n, p)$ obeys the FO 0--1 law as well.

Let $g$ be a mapping from directed graphs without loops to undirected graphs which replaces pairs of directed edges $(x, y)$ and $(y, x)$ by an undirected edges $\{x, y\}$ and deletes all the other directed edges. The mapping $g$ can be defined by the FO formula $\psi_{\sim}(x, y) = (x \rightarrow y) \wedge (y \rightarrow x)$. Then, as in the previous case, the FO sentence $g^{-1}(\varphi)$ is obtained by replacing all $x \sim y$ by $\psi_{\sim}$. Furthermore, $g(\vec{G}(n,p))\stackrel{d}=G(n,p)$ because the presence of an undirected edge $\{x, y\}$ depends only on the presence of two directed edges $(x, y)$ and $(y, x)$, and the probability that both of them are presented is $p^2$. So, we have the equality $\mathrm{Pr}(\vec{G}(n, p) \models g^{-1}(\varphi)) = \mathrm{Pr}(G(n, p^2) \models \varphi)$ which finishes the proof.
\end{proof}

Let us stress once again that in \cite{Fa} FO 0--1 laws for the binomial random graph and for $D^{\sigma}(n,p)$ are proven separately. Due to Theorem~\ref{redDigraphs}, the validity of the FO 0--1 law for $D^{\sigma}(n,p)$ for all constant $p$ implies its validity for $G(n,p)$ for all constant $p$ as well.

Actually, Theorem~\ref{redDigraphs} admits a generalisation to arbitrary signatures and distributions. We state it below, and then use for other particular reductions. It is straightforward that Theorem~\ref{redDigraphs} follows immediately from this claim.

As above, we consider two relational signatures $\sigma,\sigma'$ and respective sequences of random structures $D,D'$. For every $P\in\sigma'$ of arity $a$, assume that we are given with a FO formula $\psi_P$ of arity $a$ over $\sigma$. Let $f:\mathcal{D}\to\mathcal{D}'$ be defined as follows: for every $P\in\sigma'$ and every $X\in\mathcal{D}$, we set $f(X)\models P(x_1,\ldots,x_a)$ if and only if $X\models\psi_P(x_1,\ldots,x_a)$. We call $f$ a {\it reduction defined by $(\psi_P,\,P\in\sigma')$}. 

\begin{claim}\label{redFormula}
Let $D$ be a random structure over $\sigma$, and let $f$ be a reduction defined by $(\psi_P,\,P\in\sigma')$. Letting $D'\stackrel{d}=f(D)$, we get that $D'\preceq D$ and $f$ reduces $D'$ to $D$.
\end{claim}

\begin{proof}
In order to see that $f$ is $(\mathcal{A}\mid\mathcal{A}')$-measurable, it is sufficient to observe that, for every FO $\varphi'$ over $\sigma'$, $f^{-1}(\varphi')$ is obtained from $\varphi'$ by replacing each $P$ by $\psi_P$. Since the distribution of $D'$ coincides with the distribution of $f(D)$, we have the equality $\lim\limits_{n \to \infty}|\mathrm{Pr}(D \models f^{-1}(\varphi)) - \mathrm{Pr}(D' \models \varphi)| = 0$ which finishes the proof.
\end{proof}

A similar to $f$ object which is called {\it a FO translation} appeared in~\cite{I83} and was used for reductions between languages (for more details, we refer a reader to the book~\cite{I12}, where the concept of {\it FO reductions} is introduced and its properties are described). \\

We denote by $G_{loop}(n, p)$ the binomial random undirected graph which allows loops on the set of vertices $[n]$, i.e. every edge (out of the set of $\frac{n(n+1)}{2}$ possible edges) appears independently of the others with probability $p$. We consider this structure over the signature $\{=, \sim\}$. The proof by reduction of the FO 0--1 law for this random structure (with constant $p$) is not as straightforward as for $\vec{G}(n, p)$. Indeed, if we apply a reduction to the $D^\sigma(n, \sqrt{p})$ defined by the FO formula $P(x, y) \wedge P(y, x)$, then we get a random undirected graph which allows loops, but the probability of the presence of a loop is $p$, while a non-loop edge has the emergence probability $p^2 \not = p$. 

Let us denote by $G'_{loop}(n, p, q)$ the random undirected graph which allows loops on the set of vertices $[n]$, but over the signature $\{=, \sim', L\}$, where $\sim'$ has arity $2$ and $L$ has arity $1$. The predicate $\sim'$ expresses the presence of non-loop undirected edges, the predicate $L$ expresses the presence of loops. The distribution of $G'_{loop}(n, p, q)$ is such that each edge appears independently with probability $q$ if it is a loop, and $p$ otherwise. To prove the FO 0--1 law for $G_{loop}(n, p)$, we prove the equivalence between this structure and $G'_{loop}(n, p, p)$. The next claim is a generalisation of this statement. 

Since stochastic FO reductions define a preorder, they induce an equivalence relation on random relational structures: we call $D_n$ and $D_n'$ {\it equivalent} if $D_n \preceq D_n'$ and $D_n' \preceq D_n$. For our purposes, it is useful to have a ``loopless'' representative in every equivalence class, which is defined below. 

\begin{definition}
Let $\sigma = \{=, P_1, \dots, P_s\}$ be a signature, where $P_i$ is a predicate symbol of arity $a_i$. Let $D_n$ be a random $n$-structure over the signature $\sigma$. The random structure $D_n$ is called \textbf{loopless} if, for all $1 \leq i \leq s$ and $1 \leq j < k \leq a_i$, the following equality holds:
$$
\mathrm{Pr}(D_n \models \forall x_1 \dots \forall x_{a_i} \,\, (x_j = x_k) \Rightarrow \neg P_i(x_1, \dots, x_{a_i})) = 1.
$$
\end{definition}

\begin{claim}\label{redloopless}
For each random relational structure $D_n$ over the signature $\sigma$, there is a loopless random relational structure $D^*_n$ which is equivalent to $D_n$. Moreover, $G_{loop}(n, p)$ is equivalent to $G'_{loop}(n, p, p)$.
\end{claim}

\begin{proof}
Let us give some auxiliary definitions.
\begin{itemize}
\item Let $\mathcal{B} = \{B_1, \dots, B_t\}$ be a partition of the set $[a_i]$, for some $1 \leq t \leq a_i$.
\item Let $P_i^\mathcal{B}$ be a predicate symbol of arity $t$.
\item Let $\beta: [a_i] \to [t]$ be the mapping such that $k \in B_{\beta(k)}$.
\item Let $\beta': [t] \to [a_i]$ be a mapping such that $\beta(\beta'(k))=k$.
\item Let signature $\bar{\sigma}$ consist of $=$ and $P_i^\mathcal{B}$ over all $i$ and $\mathcal{B}$.
\item Let $\bar{\mathcal{D}}$ be the set of all relational structures over the signature~$\bar{\sigma}$.
\end{itemize}
Let $f: \mathcal{D} \to \bar{\mathcal{D}}$ be a function which maps each structure over the signature $\sigma$ to a structure over the signature $\bar{\sigma}$ by assigning each $P_i^\mathcal{B}(x_1, \dots, x_t)$ the value $P_i(x_{\beta(1)}, \dots, x_{\beta(a_i)})$, if $x_j$ are pairwise distinct, and zero, otherwise. We let the distribution of the random structure $D^*_n$ be induced by $f$ on $\bar{\mathcal{D}}$ from the probability distribution on $\mathcal{D}$.
 It is clear, that $D^*_n$ is loopless because every $P_i^\mathcal{B}(x_1, \dots, x_t)$ in the image of $f$ is zero on each tuple $(x_1, \dots, x_t)$ with two coinciding $x_j = x_k$.
By Claim~\ref{redFormula}, $f$ is a stochastic FO reduction defined by formulae
$$
P_i(x_{\beta(1)}, \dots, x_{\beta(a_i)}) \wedge \bigwedge_{1 \leq j_1 < j_2 \leq t}\neg(x_{j_1} = x_{j_2}).
$$
Also, we have a reduction $g: \bar{\mathcal{D}} \to \mathcal{D}$ defined by formulae
$$
P_i(x_1, \dots, x_{a_i}) = \bigvee_{\mathcal{B}} \left(P_i^{\mathcal{B}}(x_{\beta'(1)}, \dots, x_{\beta'(t)}) \wedge \bigwedge_{1 \leq j_1 < j_2 \leq t}\neg(x_{\beta'(j_1)} = x_{\beta'(j_2)})\right).
$$
For $G_{loop}(n, p)$, it returns the structure $G'_{loop}(n, p, p)$.
\end{proof}


Let $\tau$ be a signature $\{=, P, L\}$ where $P$ has arity $2$ and $L$ has arity $1$. There is a stochastic FO reduction $G'_{loop}(n, p, q)$ to $D^{\tau}(n, \sqrt{p}, q)$ defined by the FO formulae $P(x, y) \wedge P(y, x) \wedge \neg(x = y)$ and $L(x)$. Hence, by Corollary~\ref{corRedLaws}, Claim~\ref{redpreorder} and Claim~\ref{redloopless}, we have

\begin{theorem}
Let $p \in (0, 1)$ be a constant. The binomial undirected random graph which allows loops $G_{loop}(n, p)$ obeys the FO 0--1 law.
\label{th:0-1_loopless}
\end{theorem}

\subsubsection{Random hypergraphs}

We now switch to random $d$-uniform hypergraphs. It is known~\cite{GKLT, Fa} that oriented random hypergraphs obey the FO 0--1 law. Using the stochastic FO reduction, we show that this is also the case for unoriented random hypergraphs. Note that an unoriented hypergraph can be considered as an oriented hypergraph where all hyperedges that can be obtained from each other via a permutation $\sigma\in S_d$ are identified. Naturally, one may consider partially oriented hypergraphs by identifying hyperedges that belong to an $H$-orbit of an oriented hyperedge for some subgroup $H$ of $S_d$. Using reductions we prove the FO 0--1 law for all these hypergraphs too. Let us give an accurate definition.

Let $S_d$ be the group of permutations of $[d]$. We define the action of $S_d$ on a hyperedge $\{v_1,\ldots,v_d\}$ in the natural way: $g(v_1,\ldots,v_d) = (v_{g^{-1}(1)}, \dots, v_{g^{-1}(d)})$, $g\in S_d$. Let $H$ be a subgroup of $S_d$. 

\begin{definition}
A \textbf{$d$-uniform $H$-hypergraph} on $[n]$ is an oriented hypergraph without loops which is invariant under the action of $H$ on $[n]_d$, where $[n]_d$ is the set of all $d$-tuples of distinct vertices from $[n]$. An \textbf{$H$-hyperedge} of this graph is an orbit of its hyperedge under the action of $H$. A \textbf{binomial $d$-uniform $H$-hypergraph} $G^H(n, p)$ is a random structure over the signature $\sigma_H = \{=, P_H\}$, where $P_H$ has arity $d$, where every $H$-hyperedge is presented independently with probability $p$.
\end{definition}

A generalisation of the notion of $H$-oriented hypergraphs for structures with arbitrary number of relations was introduced in  \cite{L21} and studied in the context of convergence laws in sparse regimes. 

\begin{theorem}\label{groupRed}
Let $H$ be a subgroup of $S_d$, $K$ be a subgroup of $H$, and $p \in (0, 1)$ (not necessarily a constant). There is a stochastic FO reduction $G^H(n, p^{[H:K]}) \preceq G^K(n, p)$. In particular, for a constant $p \in (0, 1)$, $G^H(n, p)$ obeys the FO 0--1 law.
\end{theorem}
\begin{proof}
A reduction $f$ is defined by the formula
$$
\bigwedge_{g \in H}P_K(x_{g^{-1}(1)}, \dots, x_{g^{-1}(d)}).
$$
Indeed, $f(G^K(n, p))\stackrel{d}=G^H(n, p^{[H:K]})$ 
\begin{itemize}
\item
because the presence of an $H$-hyperedge $(x_1, \dots, x_d)$ in $f(G^K(n, p))$ depends only on the presence of $K$-hyperedges $(x_{g^{-1}(1)}, \dots, x_{g^{-1}(d)})$ in $G^K(n, p)$, for $g \in H$;
\item
the probability that all of them are presented is $p^{[H:K]}$ as, for each coset $gK \subseteq H$, the hyperedge $(x_{g^{-1}(1)}, \dots, x_{g^{-1}(d)})$ is presented with probability $p$, and there are exactly $[H:K]$ such cosets.
\end{itemize}

Moreover, we have a reduction $G^{\{id\}}(n, p) \preceq D^{\sigma}(n, p)$, where $\sigma=\{=, P\}$ and $P$ has an arity $d$. In the same way, as in the proof of Theorem~\ref{redDigraphs}, this reduction is defined by the formula 
$$
P(x_1, \dots, x_d) \wedge \bigwedge_{i \not = j}(x_i \not= x_j).
$$
Therefore, for each subgroup $H$ of $S_d$, there is a reduction from $G^H(n, p)$ to $D^\sigma(n, p^{|H|^{-1}})$. Since $D^\sigma(n, p^{|H|^{-1}})$ obeys the FO 0--1 law, for constant $p$, by the Corollary~\ref{corRedLaws}, $G^H(n, p)$ obeys the FO 0--1 law as well.
\end{proof}

\begin{remark} In order to give a better flavour of the FOC-hierar\-chy of random $H$-hypergraphs, we prove that conjugate subgroups of $S_d$ define the same random structure up to equivalence. 

\begin{theorem}
Let $H$ be a subgroup of $S_d$, $g \in S_d$, and $p \in (0, 1)$ (not necessarily a constant). Random structures $G^H(n, p)$ and $G^{gHg^{-1}}(n, p)$ are equivalent. 
\end{theorem}
\begin{proof}
It is sufficient to show a stochastic FO reduction from $G^H(n, p)$ to $G^{gHg^{-1}}(n, p)$ for arbitrary $H$ and $g$. The reduction $f$ is defined by the formula $\psi(x_1, \dots, x_d) = P_{gHg^{-1}}(x_{g^{-1}(1)}, \dots, x_{g^{-1}(d)})$. Let us check that this formula defines an $H$-oriented graph. Consider a permutation $h \in H$.
\begin{align*}
\psi(x_{h^{-1}(1)}, \dots, x_{h^{-1}(d)}) & = P_{gHg^{-1}}(x_{h^{-1}(g^{-1}(1))}, \dots, x_{h^{-1}(g^{-1}(d))})  \\
& = P_{gHg^{-1}}(x_{g^{-1}((ghg^{-1})^{-1}(1))}, \dots, x_{g^{-1}((ghg^{-1})^{-1}(d))})  \\ 
& = P_{gHg^{-1}}(x_{g^{-1}(1)}, \dots, x_{g^{-1}(d)}) = \psi(x_1, \dots, x_d).
\end{align*}
The penultimate equality holds true due to the definition of $P_{gHg^{-1}}$. So, permutation $g$ maps each $gHg^{-1}$-orbit of $d$-tuples to an $H$-orbit, and therefore we have that $f(G^{gHg^{-1}}(n, p))\stackrel{d}=G^H(n, p)$.
\end{proof}

\end{remark}

\subsubsection{Parametric sentences}


{\it A parametric FO sentence} is a conjunction of sentences 
$$
\forall x_1 \,\, \forall x_2 \,\, \ldots \forall x_s \bigwedge_{1 \leq i < j \leq s} (x_i \neq x_j) \Rightarrow \psi,
$$
where $\psi$ is a quantifier-free formula, and all atomic formulas $P(y_1, \dots, y_t)$ appearing in $\psi$ satisfy $\{y_1, \dots, y_t\} = \{x_1, \dots, x_s\}$. 

\begin{theorem}[Oberschelp~\cite{Oberschelp}]
\label{OberThm}
For each parametric sentence $\varphi$ over a relational signature $\sigma$, the structure $D^\sigma (n \mid \varphi)$ obeys the FO 0--1 law.
\end{theorem}

It is worth noting that the the proof from~\cite{Oberschelp} can be generalised to show the FO 0--1 law for a binomial random structure with any constant probabilities of the relations from $\sigma$ subject to a parametric sentence $\varphi$. Then, Theorem~\ref{th:0-1_loopless} is a particular case of such a generalisation.\\

For a loopless structure $D$ over $\sigma$, we say that tuples $(u_1, \dots, u_s)$ and $(v_1, \dots, v_s)$ of distinct elements of $D$ have the same {\it type} if $P(u_{\tau(1)}, \dots, u_{\tau(a)}) \Leftrightarrow P(v_{\tau(1)}, \dots, v_{\tau(a)})$ for all integer $a\geq s$, all predicates $P \in \sigma$ of arity $a$,  and all surjections $\tau:[a]\to[s]$. Let us denote by $\mathcal{P}^{\sigma}_s$ the set of all possible parametric types of $s$-element tuples in a structure over $\sigma$. Clearly, $|\mathcal{P}^{\sigma}_s|=\prod_{a\geq s}2^{|[a]\twoheadrightarrow[s]|\cdot n_a(\sigma)}$ where $|[a]\twoheadrightarrow[s]|$ is the number of all surjections $\tau:[a]\to[s]$ and $n_a(\sigma)$ is the number of predicates of arity $a$ in $\sigma$.  Each parametric sentence $\varphi$ is equivalent to a conjunction of sentences expressing that each $s$-element tuple has a parametric type in a specific subset $\mathcal{P}^{\sigma}_s(\varphi) \subset \mathcal{P}^{\sigma}_s$.
 We will denote an {\it empty parametric type}, i.e. the parametric type of an $s$-element tuple in an empty structure, by $\epsilon_s$.

\begin{claim}\label{redOber}
Let $\varphi$ be a parametric sentence. If $\epsilon_s \in \mathcal{P}^{\sigma}_s(\varphi)$ for each $s\in\mathbb{Z}_{>0}$, then $D^{\sigma}(n\mid\varphi) \preceq D^{\sigma'}(n, p_1, p_2, \dots, p_t)$ for some relational signature $\sigma'$ consisting of $t$ relational symbols and some $p_1, p_2, \dots, p_t \in [0,1]$.
\end{claim}

\begin{proof}
Applying the construction of a loopless random structure from Claim~\ref{redloopless}, we get a random structure $D^{\tilde\sigma}(n \mid \tilde\varphi)$ which is equivalent to $D^{\sigma}(n \mid \varphi)$ such that only loopless structures satisfy the parametric sentence $\tilde\varphi$. 
 Thus, further in the proof, we assume that $D^{\sigma}(n \mid \varphi)$ is loopless. On the class of loopless structures, the parametric sentence $\varphi$ is equivalent to a conjunction of sentences 
$$
\forall x_1 \,\, \forall x_2 \,\, \ldots \forall x_s \bigwedge_{1 \leq i < j \leq s} (x_i \neq x_j) \Rightarrow \psi
$$
where $\psi$ is a quantifier-free formula with predicates of arity $s$ only (and we do not require all arguments to participate in every predicate as in the case of a parametric sentence evaluated on arbitrary structures). So, $\varphi\cong_{\mathrm{loopless}}\varphi_1 \wedge \ldots \wedge \varphi_k$ where $\varphi_s$ comprises predicates of arity $s$ only and $\cong_{\mathrm{loopless}}$ denotes logical equivalence on the set of all loopless structures. Denote by $\sigma_s$ the set of relational symbols in $\sigma$ of arity $s$. Then, a loopless $\sigma$-structure $D$ satisfies $\varphi$ if and only if, for every $s\in\mathbb{Z}_{>0}$, its restriction on $\sigma_s$ satisfies~$\varphi_s$. 

Fix $s\in[k]$ and denote $m := m(\sigma,s)= |\mathcal{P}^{\sigma}_s(\varphi)|$. Consider an arbitrary ordering $\alpha_1, \dots, \alpha_{m-1}, \alpha_m = \epsilon_s$ of $\mathcal{P}^{\sigma}_s$. For each $1 \leq i < m$, we introduce a family of relational symbols $R_{s,i,1}, R_{s,i,2}, \dots, R_{s,i,r_i}$ of arity $s$ and a probability $p_{s,i} \in (0, 1)$ such that 
$$
 \left(1 - p_{s,i}(1 - p_{s,i})^{s! - 1}\right)^{r_i} = \frac{m - i}{m - i + 1}.
$$
Let us show that such $r_i$ and $p_{s,i}$ indeed exist. Fix any real $p \in (0, 1)$. Then $1 - p(1 - p)^{s! - 1} \in (0, 1)$. Therefore, there is a nonnegative integer $r_i$ such that 
$$
0<\left(1 - p(1 - p)^{s! - 1}\right)^{r_i} < \frac{m - i}{m - i + 1}.
$$
Since the function $(1 - x(1 - x)^{s! - 1})^{r_i}$ is continuous and equals $1$ at $x = 0$, there is a real number $p_{s,i}:=x \in (0, p) \subset (0, 1)$ such that $(1 - x(1 - x)^{s! - 1})^{r_i} = \frac{m - i}{m - i + 1}$. 

For $i\in[m-1]$ and $j\in[r_i]$, let $\psi_{s,i, j}(x_1, \dots, x_s)$ be the FO formula
$$
R_{s,i, j}(x_1, \dots, x_s) \wedge \bigwedge_{\tau \in S_s, \tau \neq id} \neg R_{s,i, j}(x_{\tau(1)}, \dots, x_{\tau(s)}).
$$
Let $\psi_{s,i}(x_1, \dots, x_s)$ be the FO formula 
$$
\bigvee_{1 \leq j \leq r_i} \psi_{s,i, j}(x_1, \dots, x_s).
$$

In what follows, we construct a reduction defined by formulas in such a way that, for every $i\in[m-1]$, the parametric type of $(x_1, \dots, x_s)$ over $\sigma$ is $\alpha_i$ if and only if
$$
\psi'_{s,i}(x_1, \dots, x_s) := \psi_{s,i}(x_1, \dots, x_s) \wedge \bigwedge_{1 \leq j < i} \neg \psi_{s,i}(x_1, \dots, x_s)
$$
is satisfied over $\sigma'$, and $\epsilon_s$ otherwise.

For a predicate $P \in \sigma$ of arity $s$, denote by $A_P$ the set of all $\alpha_i$ such that, for a tuple $(x_1, \dots, x_s)$ of parametric type $\alpha_i$, $P(x_1, \dots, x_s)$ holds. The desired FO formula for $P \in \sigma$ is defined as follows:
$$
\psi_{P}(x_1, \dots, x_s) = \bigvee_{\alpha_i \in A_P} \psi'_{s,i}(x_1, \dots, x_s).
$$
Due to the definition of $p_{s, i}$, we have that $\psi'_{s,i}(x_1, \dots, x_s)$ is true with the probability $\frac{1}{m}$. Therefore, formulas $\psi_{P}$ indeed define a reduction from $D^{\sigma}(n\mid\varphi)$ to the random structure $D^{\sigma'}(n,p_{s, i, j}:=p_{s,i},\,s\in[1,k], i\in[1,m(s)-1], j\in[1,r_i(s)])$
 where $\sigma'$ consists of predicates $R_{s, i, j}$.
\end{proof}

So, FO stochastic reductions allow to transfer 0--1 laws from unconditional random structures to random structures conditioned on parametric sentences. This proves Theorem~\ref{OberThm} in the particular case described in Claim~\ref{redOber}.

\begin{remark} Though we cannot show a reduction that proves Theorem~\ref{OberThm} in full, we are able to show a reasonably narrow family of random structures any random structure conditioned on a parametric sentence can be reduced to. Let $D^{\sigma}_{lord}(n, p_1, \dots, p_t)$ be a structure $D^{\sigma}(n, p_1, \dots, p_t)$ with additional predicates $LOrd_s$ (which we call {\it local orderings}) of arity $s$ for all $s$ not more than the maximal arity appearing in $\sigma$. Local orderings satisfy the following axioms:
$$
\forall x_1 \ldots \forall x_s\,\, LOrd_s(x_1, \dots, x_s) \Rightarrow \bigwedge_{1 \leq i < j \leq s} (x_i \neq x_j);
$$
$$
\forall x_1 \ldots \forall x_s \,\, LOrd_s(x_1, \dots, x_s) \Rightarrow \bigwedge_{\tau \in S_s, \tau \neq id} \neg LOrd_s(x_{\tau(1)}, \dots, x_{\tau(s)});
$$
$$
\forall x_1 \ldots \forall x_s \,\, \bigwedge_{1 \leq i < j \leq s} (x_i \neq x_j) \Rightarrow \bigvee_{\tau \in S_s}  LOrd_s(x_{\tau(1)}, \dots, x_{\tau(s)}).
$$
The respective predicates are uniformly distributed in $D^{\sigma}_{lord}(n, p_1, \dots, p_t)$ while the other $t$ predicates have corresponding probabilities $p_1,\ldots,p_t$. Note that these axioms are parametric. In  a similar way as in the proof of Theorem~\ref{OberThm}, we are able to show that, for each parametric sentence $\varphi$, there is a stochastic FO reduction $D^{\sigma}(n\mid\varphi) \preceq D^{\sigma'}_{lord}(n, p_1, p_2, \dots, p_t)$ for some relational signature $\sigma'$ and $p_1, p_2, \dots, p_t \in [0,1]$. Indeed, in the same way, we may assume that the random structure is loopless and $\varphi=\varphi_1 \wedge \ldots \wedge \varphi_k$, where $\varphi_s$ depends on predicates of arity $s$ only. Letting $m := |\mathcal{P}^{\sigma}_s(\varphi)|$, we define the probability $p_{s,i}$, $i\in[m-1]$, of the $i$th relational symbol in the new signature $\sigma'$  by $p_{s,i}^{s!} = \frac{1}{m - i + 1}.$ It remains to construct a reduction defined by formulas from $D^{\sigma}(n\mid\varphi)$ to the random structure $D^{\sigma'}_{lord}(n, p_{s, i},\,s\in[k],i\in[m(s)])$, where $\sigma'$ consists of the new predicates $R_{s,i}$, $i\in[m-1]$, and $LOrd_s$. As in the proof of Theorem~\ref{OberThm}, we do it
 in such a way that the parametric type of $(x_1, \dots, x_s)$ is the $i$th type from $\mathcal{P}^{\sigma}_s(\varphi)$, $i\in[m-1]$, if and only if
$$
\psi'_{s,i}(x_1, \dots, x_s) := \psi_{s,i}(x_1, \dots, x_s) \wedge \bigwedge_{1 \leq j < i} \neg \psi_{s,j}(x_1, \dots, x_s)
$$
is satisfied, where
$$
\psi_{s,i}(x_1, \dots, x_s)=LOrd_s(x_1, \dots, x_s) \wedge \bigwedge_{\tau \in S_s} \neg R_{s, i}(x_{\tau(1)}, \dots, x_{\tau(s)}).
$$
\end{remark}

\subsection{Application to FO limit laws for sparse structures}
\label{sc:SFOR_sparse}

For technical reasons, we will require a claim stating that a small shift of the probability parameter of a binomial random structure does not affect its equivalence class. 
 Let $\sigma$ be a relational signature, and recall that $\mathcal{D}_n$ is the set of all finite structures on $[n]$ over $\sigma$. For every $n$, consider a non-negative integer $s_n$ and an arbitrary mapping $r_n:\{0,1\}^{s_n}\to\mathcal{D}_n$. Let $B(r_n, p)$ be a random structure over the signature $\sigma$ defined as $B(r_n,p)=r_n(\xi_1,\ldots,\xi_{s_n})$, where $\xi_i$ are independent random variables with Bernoulli distribution with the parameter $p$.  Note that the binomial random graph $G(n, p)$ is distributed as $B(r_n, p)$, where $r_n$ maps a sequence of ${n\choose 2}$ ones and zeros into the graph with edges corresponding to ones in the sequence. Similarly, binomial random $d$-uniform $H$-oriented hypergraphs and binomial random oriented hypergraphs with loops are distributed as $B(r_n, p)$ with an appropriately chosen $r_n$. We shall use the following assertion about the total variation distance between Bernoulli random variables (though we believe that it might be known, we have not found it in the literature, thus we present the full proof in Appendix~\ref{AppendixE}). 

\begin{lemma}
Let $p, q \in [0, 1]$, $p_{\min} = \min\{p, q, 1 - p, 1 - q\}$, and $s_n$ be a sequence of non-negative integers. Consider random vectors $(\xi_1,\ldots,\xi_{s_n})$ and $(\eta_1,\ldots,\eta_{s_n})$, where $\xi_i$ and $\eta_i$ are independent random variables with Bernoulli distribution with the parameter $p$ and $q$ respectively. If $\lim\limits_{n \to \infty} \sqrt{\frac{s_n}{p_{\min}}} |p - q| = 0$, then the total variation distance between the distributions of $(\xi_1,\ldots,\xi_{s_n})$ and $(\eta_1,\ldots,\eta_{s_n})$ converges to zero.
\label{lm:appendix_Bernoulli}
\end{lemma}

\begin{corollary}\label{corShift}
Let $p, q \in [0, 1]$, $p_{\min} = \min\{p, q, 1 - p, 1 - q\}$, and $r_n: \{0, 1\}^{s_n} \to \mathcal{D}_n$ be a sequence of arbitrary mappings. If $\lim\limits_{n \to \infty} \sqrt{\frac{s_n}{p_{\min}}} |p - q| = 0$, then $B(r_n, p)$ is equivalent to $B(r_n, q)$.
\end{corollary}

\begin{proof}
Consider the identity mapping $id: \mathcal{D} \to \mathcal{D}$. Let us prove that it is a stochastic FO reduction. Consider a FO sentence $\varphi$. It is eligible to set $id^{-1}(\varphi):=\varphi$. Let $A_n = r_n^{-1}(\mathcal{D}(\varphi))$. Then
\begin{eqnarray*}
\mathrm{Pr}(B(r_n,p)\models\varphi) & = & \mathrm{Pr}((\xi_1,\ldots,\xi_{s_n})\in A_n),\\
\mathrm{Pr}(B(r_n,q)\models\varphi) & = & \mathrm{Pr}((\eta_1,\ldots,\eta_{s_n})\in A_n).
\end{eqnarray*}
Since the total variation distance between vertors $(\xi_1,\ldots,\xi_{s_n})$ and $(\eta_1,\ldots,\eta_{s_n})$ converges to zero, we immediately get the statement of the corollary.
\end{proof}

This corollary advances our tool to prove logical limit laws using reductions. In order to demonstrate its efficiency, we prove the following.

\begin{theorem}\label{thLogRed}
Let $r$ be a nonnegative integer.
\begin{enumerate}
\item[(i)] The random structure $D^{(r+2)}(n,(r+1) \ln n/n)$ does not obey the FO 0--1 law.
\item[(ii)] For any positive integer $k$, $D^{(r+3)}(n, (1 + \frac{1}{k} + r) \ln n/n)$ does not obey the FO 0--1 law.
\item[(iii)] For any rational $\beta \in (\frac{2}{3}, 1)$, $\mathrm{FOC}(D^{(r+3)}(n, (\beta + r) \ln n/n))$ is not totally bounded.
\end{enumerate}
\end{theorem}

\begin{proof}
We will use two FO reductions encapsulated in the two claims stated below that follow from Corollary~\ref{corShift} and reductions defined by formulae.

\begin{claim}\label{rednapr}
For any $\alpha > \frac{d}{3}$ and each nonnegative integer $r$, we have $D^{(d)}(n, \frac{c}{n^{\alpha}}) \preceq D^{(d+r)}(n, \frac{c}{n^{\alpha + r}})$.
\end{claim}

\begin{proof}
It is sufficient for us to combine Corollary~\ref{corShift} and reduction defined by a formula, to show the reduction $D^{(d)}(n, \frac{c}{n^{\alpha}}) \preceq D^{(d+1)}(n, \frac{c}{n^{\alpha + 1}})$, for $\alpha > \frac{d}{3}$. Let $P$ and $Q$ be predicates of arity $d$ and $d+1$ from signatures of $D^{(d)}(n, p)$ and $D^{(d+1)}(n, p)$ respectively. The formula $\exists y \,\, Q(y, x_1, \dots, x_d)$ reduces $D^{(d)}(n, \frac{c}{n^{\alpha}})$ to $D^{(d+1)}(n, 1 - \sqrt[n]{1 - \frac{c}{n^{\alpha}}})$. Indeed, $P(x_1, \dots, x_d)$ is true with probability 
$$
\frac{c}{n^{\alpha}} = 1 - \left(1 - \left(1 - \sqrt[n]{1 - \frac{c}{n^{\alpha}}}\right)\right)^n,
$$
and all such events are independent. Also, we have
$$
1 - \sqrt[n]{1 - \frac{c}{n^{\alpha}}} = \frac{c}{n^{\alpha + 1}} + O(n^{- 2 \alpha - 1}).
$$
Then
\begin{align*}
\sqrt{\frac{n^{d+2+\alpha}}{c}} \left|\left(1 - \sqrt[n]{1 - \frac{c}{n^{\alpha}}}\right) - \frac{c}{n^{\alpha + 1}}\right|  = \frac{1}{\sqrt{c}}n^{\frac{d+2+\alpha}{2}} O(n^{- 2 \alpha - 1})
 = O(n^{\frac{d - 3 \alpha}{2}}).
\end{align*}
Therefore, by Corollary~\ref{corShift}, we have an equivalence between $D^{(d+1)}(n, 1 - \sqrt[n]{1 - \frac{c}{n^{\alpha}}})$ and $D^{(d+1)}(n, \frac{c}{n^{\alpha + 1}})$. Note that, if $\alpha > \frac{d}{3}$, then, for each nonnegative integer $r$, we have that $\alpha + r > \frac{d+r}{3}$. Therefore, in the same way we can apply other $r-1$ reductions. It proves the claim.
\end{proof}

\begin{claim}\label{redlnnn}
For any $\alpha > d - 2$, $D^{(d)}(n, n^{-\alpha}) \preceq D^{(d+1)}(n, \frac{\alpha \ln n}{n})$.
\end{claim}

\begin{proof}
The formula $\forall y \,\,\neg Q(y, x_1, \dots, x_d)$ reduces $D^{(d)}(n, \frac{c}{n^{\alpha}})$ to $D^{(d+1)}(n, 1 - \sqrt[n]{\frac{c}{n^{\alpha}}})$. Let us estimate 
$$
1 - \sqrt[n]{\frac{c}{n^{\alpha}}} = 1 - e^{\frac{\ln c - \alpha \ln n}{n}} = \frac{\alpha \ln n - \ln c}{n} + O\left(\left(\frac{\ln n}{n}\right)^2\right).
$$

Suppose $\varepsilon$ is an arbitrary real number in the interval $(0, 1)$. For $c < 1 - \varepsilon \frac{(\ln n)^3}{n}$ and sufficiently large $n$, we have $1 - \sqrt[n]{\frac{c}{n^{\alpha}}} > \frac{\alpha \ln n}{n}$. For $c > 1 + \varepsilon \frac{(\ln n)^3}{n}$ and sufficiently large $n$, we have $1 - \sqrt[n]{\frac{c}{n^{\alpha}}} < \frac{\alpha \ln n}{n}$. Note that the function $f(c) = \sqrt[n]{\frac{c}{n^{\alpha}}}$ is continuous on the interval $(0, 2)$ and monotone. Then, for sufficiently large $n$ there is unique solution $c = c_n$ of the equation $1 - \sqrt[n]{\frac{c}{n^{\alpha}}} = \frac{\alpha \ln n}{n}$ in the interval $(0, 2)$, and $\lim\limits_{n \to \infty} (c_n - 1) \frac{n}{(\ln n)^3} = 0$. We get $D^{(d)}(n, \frac{c_n}{n^{\alpha}}) \preceq D^{(d+1)}(n, \frac{\alpha \ln n}{n})$. Thus, it remains to check the condition of Corollary~\ref{corShift} for $D^{(d)}(n, \frac{c_n}{n^{\alpha}})$ and $D^{(d)}(n, \frac{1}{n^{\alpha}})$:
$$
\sqrt{n^{d+\alpha}} \left|\frac{c_n}{n^{\alpha}} - \frac{1}{n^{\alpha}}\right| = n^{\frac{d+\alpha}{2}} o \left( \frac{(\ln n)^3}{n^{\alpha + 1}} \right) = o \left( (\ln n)^3 n^{\frac{d - 2 - \alpha}{2}} \right),
$$
completing the proof of the claim.
\end{proof}

Let us finish the proof of Theorem~\ref{thLogRed}. We start from (i). The FO 0--1 law fails for $D^{(1)}(n, \frac{1}{n})$, since the number of elements satisfying the unary predicate from the signature of this random structure converges in distribution to $\mathrm{Pois}(1)$. Using Claim~\ref{rednapr} for parameter $\alpha = 1 > \frac{1}{3} = \frac{d}{3}$, we get a reduction $D^{(1)}(n, \frac{1}{n}) \preceq D^{(r+1)}(n, \frac{1}{n^{r+1}})$. Using Claim~\ref{redlnnn} for parameter $\alpha = r + 1 > r - 1 = d - 2$, we get a reduction $D^{(r+1)}(n, \frac{1}{n^{r+1}}) \preceq D^{(r+2)}(n, \frac{(r+1) \ln n}{n})$. Then, we transfer the absence of the FO 0--1 law from $D^{(1)}(n, \frac{1}{n})$ to $D^{(r+2)}(n, \frac{(r+1) \ln n}{n})$.

Next, let us prove (ii). By Theorem~\ref{redDigraphs}, there is a reduction $G(n, (1 - n^{-1 - \frac{1}{k}})^2) \preceq D^{(2)}(n, 1 - n^{-1 - \frac{1}{k}})$. Since the inversion of all edges defines equivalences between $G(n, p)$ and $G(n, 1-p)$, and between $D^{(2)}(n, p)$ and $D^{(2)}(n, 1-p)$, we have $G(n, 2n^{-1 - \frac{1}{k}} - n^{-2 - \frac{2}{k}}) \preceq D^{(2)}(n, n^{-1 - \frac{1}{k}})$. Using Claim~\ref{rednapr} for parameter $\alpha = 1 + \frac{1}{k} > \frac{2}{3} = \frac{d}{3}$, we get a reduction $D^{(2)}(n, n^{-1 - \frac{1}{k}}) \preceq D^{(r+2)}(n, n^{-1 - \frac{1}{k}-r})$. Using Claim~\ref{redlnnn} for parameter $\alpha = 1 + \frac{1}{k} + r > r = d - 2$, we get a reduction $D^{(r+2)}(n, n^{-1 - \frac{1}{k}-r}) \preceq D^{(r+3)}(n, (1 + \frac{1}{k} + r) \ln n/n)$. Then, we obtain a reduction $G(n, 2n^{-1 - \frac{1}{k}} - n^{-2 - \frac{2}{k}}) \preceq D^{(r+3)}(n, (1 + \frac{1}{k} + r) \ln n/n)$. From \cite{SS}, we have an absence of the FO 0--1 law for $G(n, p)$ with $p \sim c n^{-1 - \frac{1}{k}}$, and then for $D^{(r+3)}(n, (1 + \frac{1}{k} + r) \ln n/n)$ as well.

Finally, we prove (iii). By Theorem~\ref{redDigraphs}, there is a reduction $G(n, 2 n^{-\beta} - n^{-2\beta}) \preceq D^{(2)}(n, n^{-\beta})$. Using Claim~\ref{rednapr} for parameter $\alpha = \beta > \frac{2}{3} = \frac{d}{3}$, we get a reduction $D^{(2)}(n, n^{-\beta}) \preceq D^{(r+2)}(n, n^{-\beta-r})$. Using Claim~\ref{redlnnn} for parameter $\alpha = \beta + r > r = d - 2$, we get a reduction $D^{(r+2)}(n, n^{-\beta-r}) \preceq D^{(r+3)}(n, (\beta + r) \ln n/n)$. Then, we obtain a reduction $G(n, 2n^{-\beta} - n^{-2\beta}) \preceq D^{(r+3)}(n, (\beta + r) \ln n/n)$ which implies that $\mathrm{FOC}(G(n, 2n^{-\beta} - n^{-2\beta})) \subseteq \mathrm{FOC}(D^{(r+3)}(n, (\beta + r) \ln n/n))$. 

The proof of Theorem~\ref{thSSFOC} and the proof of Lemma~\ref{absCLSS} in \cite{SS} do not rely on the exact equality $p = n^{-\alpha}$ but rather works for every $p \sim c n^{- \alpha}$. Then, $\mathrm{FOC}(G(n, 2n^{-\beta} - n^{-2\beta}))$ is not totally bounded, and it finishes the proof of the theorem by Claim~\ref{redFOC}.
\end{proof}

\subsection{First order complexity of random $H$-hypergraphs around the connectivity threshold}
\label{sc:SFOR_Muller}

For a subgroup $H$ of $S_d$, the random hypergraph $G^H(n,p)$ is defined in Section 3.2. In particular, $G^{S_d}(n,p)$ is the well-studied and commonly considered binomial unoriented hypergraph. Larrauri, M\"{u}ller, and Noy~\cite{LMN} proved the following.
\begin{theorem}[Larrauri, M\"{u}ller, Noy~\cite{LMN}]\label{thmLMN}
The set of limits $\lim_{n\to\infty}\mathrm{Pr}(G^{S_d}(n, \frac{c}{n^{d-1}})\models\varphi)$ over all FO sentences $\varphi$ is dense in $[0,1]$ if and only if $c\geq c_0^{S_d}$, where $c_0^{S_d}$ is the unique positive solution of the equation 
$$
\frac{1}{2} \ln \frac{1}{1 - \frac{c}{(d-2)!}} - \frac{c}{2(d-2)!} = \ln 2.
$$
\end{theorem}
Using stochastic FO reductions, we generalise this result to all possible orientations. \\

In \cite{L21} it was proven that $G^H(n, \frac{c}{n^{d-1}})$ obeys the FO convergence law. For $c>0$, let $L_c^H$ be the set of $\lim\limits_{n \to \infty} \mathrm{Pr}(G^H(n, \frac{c}{n^{d-1}}) \models \varphi)$ over all FO sentences $\varphi$. Let us denote by $c_0^H$ the infimum of the set of positive numbers $c$ such that $L_c^H$ is dense in $[0,1]$. In this section, we prove the following.

\begin{theorem}\label{generalLMN}
Let $H$ be a subgroup of $S_d$, $c>0$, and let $d \geq 2$ be an integer. Then, $c_0^H = \frac{|H|}{d!}c_0^{S_d}$. Moreover, for $c \geq c_0^H$, $L_c^H$ is dense.
\end{theorem}

\begin{proof}
Theorem~\ref{generalLMN} follows from Claims~\ref{claim9}~and~\ref{CL:CLAIM10} stated below, their proofs appear after the proof of Theorem~\ref{generalLMN}.

\begin{claim}\label{claim9}
Let $H$ be a subgroup of $S_d$, $K$ be a subgroup of $H$, $c>0$, and let $d \geq 2$ be an integer. Then, $c_0^H \geq [H:K]c_0^K$. Moreover, for any $c \geq \frac{|H|}{d!}c_0^{S_d}$, the set $L_c^H$ is dense.
\end{claim}

\begin{claim}\label{CL:CLAIM10}
For $c < c_0^{S_d}/d!$, the set $L_c^{\{id\}}$ is not dense in $[0,1]$.
\end{claim}

By Claim \ref{CL:CLAIM10}, $c_0^{\{id\}} \geq c_0^{S_d}/d!$. Combining with Claim \ref{claim9}, we get $\frac{|H|}{d!}c_0^{S_d} \geq c_0^{H} \geq |H| c_0^{\{id\}} \geq \frac{|H|}{d!}c_0^{S_d}$. Thus, $c_0^H=\frac{|H|}{d!}c_0^{S_d}$. Moreover, due to Claim~\ref{claim9}, for each $c \geq c_0$, $L_c^H$ is dense, completing the proof of Theorem~\ref{generalLMN}.

\end{proof}

\begin{proof}[Proof of Claim~\ref{claim9}.]
Note that $G^H(n, p)$ is equivalent to $G^H(n, 1 - p)$ due to reductions defined by negations. From Theorem~\ref{groupRed} and Corollary~\ref{corShift}, we have the following reductions
\begin{align*}
G^H(n, [H:K]cn^{1-d}) & \preceq G^H(n, 1 - [H:K]cn^{1-d}) \\
 & \preceq G^H(n, (1 - cn^{1-d})^{[H:K]})
 \preceq G^K(n, 1 - cn^{1-d}) \preceq G^K(n, cn^{1-d}).
\end{align*}
We immediately get
\begin{proposition}\label{reduceH}
$G^H(n, [H:K]c/n^{d-1}) \preceq G^K(n, c/n^{d-1})$. 
\end{proposition}
By Proposition~\ref{reduceH}, for all $c>0$, $L_{[H:K]c}^H \subseteq L_c^K$. Therefore, if $L^H_{[H:K]c}$ is dense in $[0,1]$, then $L^K_c$ is dense in $[0,1]$ as well. Hence, $c_0^H \geq [H:K]c_0^K$. Applying this inequality for $S_d$ and $H$, we get $\frac{|H|}{d!}c_0^{S_d} = [S_d:H]^{-1}c_0^{S_d} \geq c_0^H$. Also, for each $c \geq \frac{|H|}{d!}c_0^{S_d}$, $L_{d!c/|H|}^{S_d}$ is dense in $[0,1]$ by Theorem~\ref{thmLMN}. By Proposition~\ref{reduceH}, $G^{S_d}(n, d!c/(|H|n^{d-1})) \preceq G^H(n, c/n^{d-1})$. Therefore, the set $L_c^H$ is dense in $[0,1]$ as well, completing the proof.
\end{proof}


In order to prove Claim~\ref{CL:CLAIM10}, we apply the strategy that was used in~\cite{LMN} to derive the analogous statement for {\it unoriented} hypergraphs. The proof is based on a Poisson limit theorem for the number of small cycles in random hypergraphs and the validity of the FO 0--1 law subject to the absence of cycles. Let us first recall these auxiliary results.

\begin{claim}[Larrauri, M\"{u}ller, Noy~\cite{LMN}]\label{hc}
Let $p \sim \frac{c}{n^{d-1}}$ with $c>0$. Set
$$
f(c) = \sum_{k \geq 2} \frac{\left(\frac{c}{(d-2)!}\right)^k}{2k} = \frac{1}{2} \ln \frac{1}{1 - \frac{c}{(d-2)!}} - \frac{c}{2(d-2)!}.
$$
Let $X_n$ be the total number of cycles in $G^{S_d}(n, p)$. Then 
$$
\lim\limits_{n \to \infty}\mathbb{E}(X_n) = f(c),
$$
$$
\lim\limits_{n \to \infty} \mathrm{Pr}(X_n = 0) = e^{-f(c)} = e^{\frac{c}{2(d-2)!}} \sqrt{1 - \frac{c}{(d-2)!}}.
$$
\end{claim}

We define a cycle in an $d$-uniform oriented hypergraph as an oriented hypergraph with a set $W$ of $(d-1)s$ vertices and $s$ hyperedges such that there is no proper $W' \subset W$ inducing at least $\geq \frac{|W'|}{d-1}$ hyperedges in this hypergraph. Let us show that Claim~\ref{hc} implies

\begin{corollary}\label{dhc}
Let $p \sim \frac{c}{n^{d-1}}$ with $c>0$. Set
\begin{align}
f(c) = \sum_{k \geq 2} \frac{\left(d(d-1)c\right)^k}{2k} = \frac{1}{2} \ln \frac{1}{1 - d(d-1)c} - \frac{d(d-1)c}{2}. \label{eq:1Cor3}
\end{align}
Let $X_n$ be the total number of cycles in $G^{\{id\}}(n, p)$. Then 
$$
\lim\limits_{n \to \infty}\mathbb{E}(X_n)=f(c),
$$
$$
\lim\limits_{n \to \infty} \mathrm{Pr}(X_n = 0) = e^{-f(c)} = e^{\frac{d(d-1)c}{2}} \sqrt{1 - d(d-1)c}.
$$
\end{corollary}
\begin{proof}
For each oriented hypergraph $G$, we consider the unoriented hypergraph $G'$ obtained by erasing orientations: there is a hyperedge $\{v_1, \dots, v_d\}$ in $G'$ if and only if there is an oriented edge $(v_{\sigma(1)}, \dots, v_{\sigma(d)})$ in $G$, for some $\sigma \in S_d$. The corollary follows from the next two observations. 
\begin{itemize}
\item
For $G \sim G^{\{id\}}(n, p)$, the number of cycles in $G$ equals to the number of cycles in $G'$ whp, since whp $G$ does not contain two different hyperedges on the same set of $d$ vertices.
\item
If $G \sim G^{\{id\}}(n, p)$, then $G' \sim G^{S_d}(n, p')$, where $p' \sim \frac{d!c}{n^{d-1}}$. 
\end{itemize}

Then, we conclude that the number of cycles in $G^{\{id\}}(n, p)$ has the same distribution as the number of cycles in $G^{S_d}(n, p')$ asymptotically. Then, replacing of $c$ by $cd!$ in the statement of Claim~\ref{hc} yields Corollary~\ref{dhc}.

\end{proof}

Finally, let us state the validity of the FO 0--1 law subject to the absence of cycles. 

\begin{lemma}\label{lemmaTrees}
Let $p \sim \frac{c}{n^{d-1}}$ with $0<c<\frac{1}{d(d-1)}$. Let $X_n$ be the random variable equal to the total number of cycles in $G^{\{id\}}(n, p)$. Let $\varphi$ be a FO sentence. Then,
$$
\lim\limits_{n \to \infty} \mathrm{Pr}(G^{\{id\}}(n, p) \models \varphi\mid X_n = 0) \in \{0, 1\}.
$$
\end{lemma}

The proof is literally the same as the proof of Lemma 4.7 in~\cite{LMN} that states a similar 0--1 law for unoriented hypergraphs. It is a direct corollary of the fact that the number of connected components in $G^{\{id\}}(n, p)$ isomorphic to any fixed tree is not bounded in probability. So, we omit this proof.

\begin{proof}[Proof of Claim \ref{CL:CLAIM10}]
Note that $\frac{1}{d!}c_0^{S_d}$ is the unique positive solution of $e^{-f(c)}=\frac{1}{2}$, where $f$ is defined in~(\ref{eq:1Cor3}). By Corollary~\ref{dhc}, for each $c < c_0$, we have that 
$$
\lim\limits_{n \to \infty} \mathrm{Pr}(X_n = 0) = e^{-f(c)} > \frac{1}{2}.
$$
By Lemma~\ref{lemmaTrees}, we have that 
$$
\lim\limits_{n \to \infty} \mathrm{Pr}\left(G^{\{id\}}\left(n, \frac{c}{n^{d-1}}\right) \models \varphi \right) \in [0, 1 - e^{-f(c)}] \cup [e^{-f(c)}, 1].
$$
Then, the set $L_c^{\{id\}}$ is not dense in $[0,1]$. 
\end{proof}

\section{Decision problem}
\label{sc:rec_enum}

In this section, we prove that the problem of determining whether, given an input FO sentence $\varphi$, $G(n\mid\varphi)$ obeys the FO 0--1 law is not recursively enumerable. In what follows, we refer to this decision problem as $\mathsf{0}$--$\mathsf{1LAW}$.

\begin{theorem}
$\mathsf{0}$--$\mathsf{1LAW}$ is not recursively enumerable.
\label{th:RE}
\end{theorem}

\begin{proof}
It is enough to reduce the problem of deciding whether a Diophantine equation has a solution ($\mathsf{DE}$ in what follows) to the complement of $\mathsf{0}$--$\mathsf{1LAW}$.  Indeed, $\mathsf{DE}$ is complete in the class of recursively enumerable languages due to Matiiasevich \cite{M70} (see also~\cite{M93} for a survey on the negative solution of Hilbert's 10th problem), then its complement $\overline{\mathsf{DE}}$ is not recursively enumerable, and therefore from the reduction we would get that $\mathsf{0}$--$\mathsf{1LAW}$ is not recursively enumerable as well.

So, for each integer polynomial $P(x_1, \dots, x_k)$ we shall compute a FO sentence $\varphi_P$ such that $P(x_1, \dots, x_k) = 0$ has an integer solution if and only if $G(n\mid\varphi_P)$ does not obey the FO 0--1 law. 

We first reduce the problem for integer solutions to a problem for positive integer solutions. For integer polynomial $P(x_1, \dots, x_k)$, there is a solution of $P(x_1, \dots, x_k) = 0$ in integers if and only if $P(y_1-z_1, \dots, y_k-z_k) = 0$ has a solution in positive integers. Also, we can move all monomials with negative coefficients in $P(y_1-z_1, \dots, y_k-z_k) = 0$ to the right-hand side of the equality. Then, we get an equation $Q(y_1, z_1, \dots, y_k, z_k) = R(y_1, z_1, \dots, y_k, z_k)$, where $Q$ and $R$ have nonnegative integer coefficients. 

Let us now consider a system of equations $\mathcal{S}$ that has a solution in positive integers if and only if $Q=R$ has such a solution. Initially we denote occurrences of all variables in $P$ or $Q$ by $t_1,\ldots,t_s$ --- here $s$ is the total number of occurrences of variables, taking into account their powers. Then for every $i\in[2,s]$, we add $t_i=t_j$ to the system, if there exists $j<i$ such that $t_i$ and $t_j$ denote exactly the same variable. Then, for each of the two polynomials $Q$ and $P$, we consider the sequence of arithmetic operations that are applied to compute them. Observe that the $j$-th operation can be written as either $t_{s+j}=t_it_{i'}$ or $t_{s+j}=t_i+t_{i'}$ for certain $i,i'\leq s+j-1$. For each of the operations, we add the respective equation to the system. Let us assume that $Q$ and $P$ are computed at steps $q$ and $p$ respectively (that is, $t_q=Q$ and $t_p=P$). The last equation in the system is $t_q=t_p$. Let us observe that indeed the initial Diophantine equation has a solution in integers if and only if the constructed system of equations $\mathcal{S}$ has a solution in positive integers.

Sequence of computations of $P$ and $Q$ encoded in $\mathcal{S}$ can be represented also by a FO sentence $\psi_P$ with the following properties: (1) $\psi_P$  has finite models if and only if $P=Q$ has solutions in positive integers; (2) if $G\models\psi_P$, then $G$ has even number of vertices and a graph obtained from $G$ by the addition of an isolated edge satisfies $\psi_P$ as well. Let us construct the sentence $\psi_P$ explicitly. First of all, we consider
\begin{itemize}
\item
a FO formula $\mathrm{Deg}_{\geq d}(v)$ saying that $v$ has degree at least $d$;
\item
a FO formula $\mathrm{Leaf}_i(v)$ saying that $v$ has exactly $2i+1$ neighbours of degree~$1$;
\item
$\mathrm{All}=\forall v \,\, \mathrm{Deg}_{\geq 1}(v) \wedge \left(\mathrm{Deg}_{\geq 2}(v) \Rightarrow \bigvee\limits_{i = 1}^s \mathrm{Leaf}_i(v)\right).$
\end{itemize}
The sentence $\mathrm{All}$ divides vertices into $s+1$ types: vertices of the first type have degree 1, vertices of type $d\in[2,\ldots,s+1]$ have exactly $2d+1$ neighbours of degree $1$. We then encode the three types of equations from $\mathcal{S}$ by FO sentences in the following way. 
\begin{itemize}

\item For $1\leq i<j\leq s$, $\mathrm{Equal}_{i,j}$ states that for each vertex $v$ that satisfies $\mathrm{Leaf}_i$ there is a unique vertex $u$ such that $\mathrm{Leaf}_j(u)\wedge(u \sim v)$, and vice versa. That is, there is a perfect matching between the sets of vertices that satisfy $\mathrm{Leaf}_i$ and $\mathrm{Leaf}_j$ which forces the numbers of these vertices to be equal.

\item For $1\leq i<j<d\leq s$, $\mathrm{Sum}_{i,j,d}$ states that for each vertex $v$ that satisfies $\mathrm{Leaf}_d$ there is a unique vertex $u$ such that $(\mathrm{Leaf}_i(u) \vee \mathrm{Leaf}_j(u))\wedge(u \sim v)$, and vice versa. That is, there is a perfect matching between the sets of vertices that satisfy $\mathrm{Leaf}_d$ and $\mathrm{Leaf}_i \vee \mathrm{Leaf}_j$ which forces the respective cardinalities to be equal.

\item For $1\leq i<j<d\leq s$,  $\mathrm{Prod}_{i,j,d}$ states that for each vertex $v$ that satisfies $\mathrm{Leaf}_d$ there is a unique pair of vertices $u,w$ such that $\mathrm{Leaf}_i(u) \wedge \mathrm{Leaf}_j(w)\wedge(u \sim v) \wedge (u \sim w)$, and vice versa. That is, there is a bijection between the set of vertices that satisfy $\mathrm{Leaf}_d$ and the set of pairs of vertices $u,w$ such that $\mathrm{Leaf}_i(u) \wedge \mathrm{Leaf}_j(w)$ which forces the cardinality of the first set to be equal to the product of cardinalities of sets $\{u\mid \mathrm{Leaf}_i(u)\}$ and $\{w\mid \mathrm{Leaf}_j(w)\}$.
\end{itemize}

Finally 
$$
\psi_P:=\mathrm{All}\wedge\bigwedge_{t_i=t_j} \mathrm{Equal}_{i,j}\wedge\bigwedge_{t_d=t_i+t_j} \mathrm{Sum}_{i,j,d}\wedge\bigwedge_{t_d=t_i\cdot t_j} \mathrm{Prod}_{i,j,d}
$$
where the conjunctions are over the respective equations in $\mathcal{S}$. 
Note that if $G\models\psi_P$ for a certain $G$, then the system $\mathcal{S}$ has a solution $(t_1,\ldots,t_s)$ in positive integers: $t_i$ is the numbers of vertices in $G$ that satisfy $\mathrm{Leaf}_i$. On the other hand if $\mathcal{S}$ has a solution $(t_1,\ldots,t_s)\in\mathbb{Z}_{>0}^s$, then a graph $G$ that satisfies $\psi_P$ can be constructed directly: first, consider a disjoint union of sets of vertices $T_1,\ldots,T_s$, $|T_i|=t_i$. Then, for every $i\in[s]$, join $2i+1$ leaves to every vertex from $T_i$. Next, observe that every successive equation from $\mathcal{S}$, but the last one, contains a new variable --- thus, we may recursively draw edges in the desired way defined by all equations, but the last one. Finally, if we get a graph that contains at least one edge between $T_p$ and $T_q$, then it may only happen if a perfect matching between $T_p$ and $T_q$ is already drawn since, otherwise, $|T_p|\neq |T_q|$. On the other hand, if no edge has been added between $T_p$ and $T_q$, then we add a perfect matching between them and finish the construction of $G$. 

Let us finally observe that any graph that satisfies $\psi_P$ must have even number of vertices: there is even number of vertices involved in isolated edges and any vertex satisfying $\mathrm{Leaf}_i$ has $2i+1$ neighbours of degree one.

We define the desired sentence $\varphi_P$ as  $\mathrm{Empty}\vee \psi_P$, where $\mathrm{Empty}=\forall x\forall y\,\neg(x\sim y)$ describes the property of being empty. If there are {\it no} integer solutions of $P(x_1, \dots, x_k) = 0$, then there are no graphs satisfying $\psi_P$. Therefore, any $G$ satisfying $\varphi_P$ is empty. We immediately get that $G(n\mid\varphi_P)$ obeys the FO 0--1 law. On the other hand, if there is an integer solution of $P(x_1, \dots, x_k) = 0$, then there is a solution  of $\mathcal{S}$ in $\mathbb{Z}_{>0}^s$. Let $G_0:=G=G(\mathcal{S})$ be the graph satisfying $\psi_P$ as above. Let $G_i$ be obtained from $G_0$ by adding $i$ isolated edges. Obviously, $G_i\models\psi_P$ for all $i$, and  $|V(G_i)| = 2i + |V(G_0)|$. For odd $n$, $\mathrm{Pr}(G(n\mid\varphi_P)\models\mathrm{Empty})=1$ because there are no graphs with odd number of vertices satisfying $\psi_P$. In contrast, for all even $n \geq |G_0|$, this probability is at most $1/2$, because there is at least one nonempty graph satisfying $\phi_P$. Therefore,  $\mathrm{Pr}(G(n\mid\varphi_P) \models\mathrm{Empty})$ does not converge, and so $G(n\mid\varphi_P)$ does not obey the FO 0--1 law, completing the proof.
\end{proof}

\begin{remark}
Actually, we proved that $\mathsf{0}$--$\mathsf{1LAW}$ is $\Pi_1$-hard. On the other hand, $\varphi \in \mathsf{0}$--$\mathsf{1LAW}$ if and only if
$$
\forall M \forall\psi\,\, \exists N\,\, \forall n > N\,\, \mathrm{Pr}(G(n \mid \varphi) \models \psi)\notin(1/M,1-1/M).
$$
Since the property $\mathrm{Pr}(G(n \mid \varphi) \models \psi)\notin(1/M,1-1/M)$ is decidable, we have that  $\mathsf{0}$--$\mathsf{1LAW} \in \Pi_3$. Unfortunately, we do not manage to find the level of $\mathsf{0}$--$\mathsf{1LAW}$ in the arithmetical hierarchy.
\end{remark}

\section{Proof of part (iii) of Theorem~\ref{TH:AXFOC}}
\label{sc:proofiii}

This section is divided into four parts. In Section~\ref{preliminaries}, we recall the necessary background used in the proof. In Section~\ref{overview} we show the general scheme of the proof and reduce the theorem to a construction of two FO sentences $\varphi_1$ and $\varphi_2$ that have to satisfy certain properties. After that, in Section~\ref{varphi1} and Section~\ref{varphi2} we construct the desired FO sentences and verify their properties.

\subsection{Preliminaries}\label{preliminaries}

Graphs $G$ and $H$ are {\it (elementary) $k$-equivalent} (we write $G \cong_k H$) if there is no FO sentence of {\it quantifier depth} $k$ that distinguishes between $G$ and $H$ (quantifier depth is the maximum number of nested quantifiers in the sentence, see the formal definition in~\cite{Libkin}). We have to recall a combinatorial approach to proving the elementary equivalence, this tool is widely known as the Ehrenfeucht--Fra\"{\i}ss\'{e} game. For simplicity of presentation (and this is also enough for our purposes), we recall the definition of this game for graphs; for arbitrary relational structures see \cite{Eh, Fr, Libkin}.

\begin{definition}
In the {\it Ehrenfeucht--Fra\"{\i}ss\'{e} game} on graphs $G$ and $H$ with $k$ rounds, there are two players, {\it Spoiler} and {\it Duplicator}. In round $i\in[k]$, Spoiler chooses a vertex in one of the graphs and then Duplicator chooses a vertex in the other graph. After  $k$ rounds are played, there are two tuples of $k$ not necessarily different vertices $(x_1, \dots, x_k)$ chosen in $G$ and $(y_1, \dots, y_k)$ chosen in $H$ are chosen by both players (the $i$-th vertex is chosen in the $i$-th round). Duplicator wins if $(x_i = x_j) \Leftrightarrow (y_i = y_j)$ and $(x_i \sim x_j) \Leftrightarrow (y_i \sim y_j)$, for each pair $1 \leq i < j \leq k$, that is the map $x_i\to y_i$ is a partial isomorphism between $G$ and $H$.
\end{definition}

\begin{theorem}[Ehrenfeucht, Fra\"{\i}ss\'{e} \cite{Eh, Fr}]
For two graphs $G$ and $H$, Duplicator has a winning strategy in the Ehrenfeucht--Fra\"{\i}ss\'{e} game on graphs $G$ and $H$ in $k$ rounds if and only if $G \cong_k H$.
\end{theorem}

For constructing winning strategies in the Ehrenfeucht--Fra\"{\i}ss\'{e} game, we need the following auxiliary simple claims. 

\begin{claim}\label{Cartesian}
Let Spoiler has no winning strategy in $k$ moves on graphs $G_1$ and $G_2$. Then, for each graph $H$, there is no winning strategy for Spoiler in $k$ moves on graphs $G_1 \Box H$ and $G_2 \Box H$.
\end{claim}

\begin{proof}
There is a winning strategy for Duplicator in $k$ moves on $G_1$ and $G_2$. Let us construct a winning strategy for $G_1 \Box H$ and $G_2 \Box H$. Suppose in the first $t-1$ rounds vertices $(u_1, v_1), \dots, (u_{t-1}, v_{t-1})$ from $G_1 \Box H$ and $(u'_1, v'_1), \dots, (u'_{t-1}, v'_{t-1})$ from $G_2 \Box H$ are chosen. Without loss of generality, suppose that Spoiler chooses $(u_t, v_t) \in G_1 \Box H$ in round $t$. Let $u'_t \in G_2$ be the Duplicator's choice in round $t$ according to the strategy on $G_1$ and $G_2$ when $u_1, \dots, u_t \in G_1$ and $u'_1, \dots, u'_{t - 1}$ are chosen. Then, Duplicator chooses $(u'_t, v_t)$ in the game on $G_1 \Box H$ and $G_2 \Box H$. 

Following this strategy, we have that $v'_i = v_i$, for all $i \leq k$ after $k$ rounds. Also, vertices $u_1, \dots, u_k \in G_1$ and $u'_1, \dots, u'_k \in G_2$ are chosen according to the winning strategy of Duplicator in the game on $G_1, G_2$. Therefore, 
\begin{align*}
(u_i, v_i) \sim (u_j, v_j)
& \Leftrightarrow ((u_i \sim u_j) \wedge (v_i = v_j)) \vee ((u_i = u_j) \wedge (v_i \sim v_j))  \\
& \Leftrightarrow ((u'_i \sim u'_j) \wedge (v'_i = v'_j)) \vee ((u'_i = u'_j) \wedge (v'_i \sim v'_j))\\
& \Leftrightarrow (u'_i, v'_i) \sim (u'_j, v'_j),
\end{align*}
so indeed Duplicator wins the game on $G_1 \Box H$ and $G_2 \Box H$.
\end{proof}

\begin{claim}\label{glue}
Let $Y$, $G$, $H$ be a triple of graphs with induced subgraphs $X_0 \subset Y$, $X_1 \subset G$, $X_2 \subset H$ and isomorphisms $f: X_0 \to X_1$, $g: X_0 \to X_2$. Denote by $G_X$ and $H_X$  relational structures with the usual adjacency relation in $G$ and $H$ and a unary relation $I_x$ for each vertex $x \in X_0$ such that $I_x(y) \Leftrightarrow (y = f(x))$ for $y \in G$ and $I_x(z) \Leftrightarrow (z = g(x))$ for $z \in H$. Let $G\cup_f Y$ be the graph obtained from $G$ by identifying every $x \in X_0$ with $f(x) \in X_1$ and adding the rest of $Y$ vertex-disjointly. Similarly define $H\cup_g Y$. Assume that Duplicator has a winning strategy in $k$ moves on $G_X$ and $H_X$, then Duplicator has a winning strategy in $k$ moves on $G\cup_f Y$ and $H\cup_g Y$ as well.
\end{claim}

\begin{proof}
Denote by $Y_1 \subset G\cup_f Y$ and $Y_2 \subset H\cup_g Y$ the two copies of $Y$ attached to $G$ and $H$ respectively, and let $h: Y_1 \to Y_2$ be an isomorphism between them such that its restriction to $X_1$ coincides with $g\circ f^{-1}$. Suppose that in first $t-1$ rounds vertices $u_1, \dots, u_{t-1} \in G\cup_f Y$ and $u'_1, \dots, u'_{t-1} \in H\cup_g Y$ are chosen. Without loss of generality, Spoiler chooses a vertex $u_t \in G\cup_f Y$ in round $t$. If $u_t \in G$, Duplicator chooses the vertex $u'_t \in H$ according to a winning strategy on $G_X$ and $H_X$. If $u_t \in Y_1$, Duplicator chooses $h(u_t) \in Y_2$. Note that this strategy is well-defined since for $u_t \in X_1 = G \cap Y_1$, the vertex chosen according to the winning strategy on $G_X$ and $H_X$ has to be $h(u_t)$ because of the unary predicate $I_{f^{-1}(u_t)}$.

Let us prove that this is a winning strategy of Duplicator. Consider two chosen pairs of vertices $u_i, u_j \in G\cup_f Y$ and $u'_i, u'_j \in H\cup_g Y$. If $u_i, u_j \in G$, then $u'_i, u'_j \in H$ and $(u_i \sim u_j) \Leftrightarrow (u'_i \sim u'_j)$ since this vertices were chosen according to the winning strategy of Duplicator in the game on $G_X$ and $H_X$. If $u_i, u_j \in Y_1$, then $u'_i, u'_j \in Y_2$ and $(u_i \sim u_j) \Leftrightarrow  (u'_i \sim u'_j)$ because $u'_i = h(u_i)$ and $u'_j = h(u_j)$. If $u_i \in G \backslash Y_1, u_j \in Y_1 \backslash G$, then $u'_i \in H \backslash Y_2, u'_j \in Y_2 \backslash H$ and so there are no edges $\{u_i, u_j\}, \{u'_i, u'_j\}$. Therefore, Duplicator wins.
\end{proof}

\begin{claim}\label{leaves}
Let $G$, $H$ be graphs with an additional unary relation $I$. Let $G_I$, $H_I$ be graphs obtained from $G$,  $H$ by attaching exactly one leaf to each vertex $x$ satisfying $I(x)$. Assume that Duplicator has a winning strategy in $k$ rounds on $G$ and $H$ equipped with $I$, then Duplicator has a winning strategy in $k$ rounds on $G_I$ and $H_I$ as well.
\end{claim}

\begin{proof}
Let $f_G: V(G_I) \to V(G)$ be a mapping that acts as identity on vertices of $G$ and sends each vertex in $V(G_I) \backslash V(G)$ to its unique neighbour. Similarly, define $f_H$. Suppose that in the first $t-1$ rounds vertices $u_1, \dots, u_{t-1} \in G_I$ and $u'_1, \dots, u'_{t-1} \in H_I$ are chosen. Without loss of generality, Spoiler chooses a vertex $u_t \in G_I$ in round $t$. Let $v'_t \in H$ be the Duplicator's choice in round $t$ according to the strategy on $G$ and $H$ with unary relation $I$ when $f(u_1), \dots, f(u_t) \in G$ and $f(u'_1), \dots, f(u'_{t - 1})\in H$ are chosen. If $u_t \in G$, then Duplicator chooses $u'_t = v'_t$. If $u_t \not \in G$, then Duplicator chooses the leaf $u'_t$ attached to $v'_t$. Note that in both cases $f(u'_t)=v'_t$, so this strategy is well-defined. 

Let us prove that this strategy is winning for Duplicator. Since $f(u'_t)$ is chosen according to the winning strategy in the game on $G$ and $H$ with unary relation $I$, $f(u_i) \sim f(u_j) \Leftrightarrow f(u'_i) \sim f(u'_j)$ and $f(u_i) = f(u_j) \Leftrightarrow f(u'_i) = f(u'_j)$, for $1 \leq i < j \leq k$. Also, $u_i \in G \Leftrightarrow u'_i \in H$ for $1 \leq i \leq k$. Therefore,
\begin{align*}
 (u_i \sim u_j) & \Leftrightarrow ((u_i, u_j \in G) \wedge (f(u_i) \sim f(u_j)))\vee ((u_i \neq u_j) \wedge (f(u_i)= f(u_j))) \\
& \Leftrightarrow ((u'_i, u'_j \in H) \wedge (f(u'_i) \sim f(u'_j))) \vee  ((u'_i \neq u'_j) \wedge (f(u'_i)= f(u'_j))) \Leftrightarrow (u'_i \sim u'_j),
\end{align*}
so indeed Duplicator wins the game on $G_I$ and $H_I$.
\end{proof}

Moreover, we need a well-known fact that FO sentences express only ``local'' properties in the following sense.

\begin{theorem}[Hanf locality theorem~\cite{Ha}]\label{hanf}
For a positive integer $r$ and a graph $G$, denote by $n_r(x, G)$ the number of vertices $x'$ in $G$ such that $r$-neighbourhoods of $x$ and $x'$ are isomorphic. For every positive integer $k$ there exist positive integers $r(k)$ and $s(k)$ such that the following implication holds true for any two graphs $G$ and $H$. If, for every vertex $x \in G \sqcup H$, $\min\{s(k), n_r(x, G)\} = \min\{s(k), n_r(x, H)\}$, then $G \cong_k H$.
\end{theorem}

Finally, we make use of the following theorem about the distribution of the number of cycles in a uniformly random permutation.

\begin{theorem}[\cite{Ko}, Theorem 8]
\label{cyclesDistr}
Let $n$ be a positive integer, and let $Y_i$ be the number of cycles of length $i$ in a uniformly random permutation from $S_n$. Then, $(Y_1, \dots, Y_{n}) \stackrel{d}{\to} (\xi_1, \dots, \xi_{n})$, where $\xi_i$ are independent $\mathrm{Pois}(1/i)$ random variables.
\end{theorem}

\subsection{Overview of the proof and basic properties of $\varphi$}\label{overview}
We define the desired sentence $\varphi$ as a disjunction of three sentences $\varphi_0$, $\varphi_1$ and $\varphi_2$, such that the following holds:
\begin{enumerate}
\item[(a)] any graph satisfies at most one sentence from $\varphi_0$, $\varphi_1$ and $\varphi_2$;
\item[(b)] $\varphi_0$ expresses the property of graph being empty;
\item[(c)] for each $n$ divisible by $6$, there are $n!$ graphs on $[n]$ that satisfy $\varphi_1$, while, for any $n$ which is not divisible by $6$, any graph on $[n]$ does not satisfy $\varphi_1$;
\item[(d)] for each FO $\psi$, the limit $\lim\limits_{m\to\infty}\mathrm{Pr}(G(6m\mid\varphi_1)\models \psi)$ exists;
\item[(e)] there is a family of FO sentences $\varphi_{2,d}$, $d\geq 3$, such that $\varphi_2$ and $\bigvee_{d\geq 3}\varphi_{2,d}$ are not distinguished by any finite graph, and $\varphi_{2,d}\wedge\varphi_{2,d'}$ are contradictions for all distinct $d,d'$;
\item[(f)] for every $d\geq 3$ and every $n=6dm^2$ for some integer $m\neq 0$, there are exactly $n!/d!$ graphs on $[n]$ that satisfy $\varphi_{2,d}$, while for all other $n$, any graph on $[n]$ does not satisfy $\varphi_{2,d}$;
\item[(g)] for all FO $\psi$ and $d\geq 3$, $\lim\limits_{m\to\infty}\mathrm{Pr}(G(6dm^2\mid\varphi_{2,d})\models \psi)$ is either 0 or 1.
\end{enumerate}

Let us first show that these seven assumptions are enough for obtaining the desired properties of $\mathrm{FOC}(G(n\mid\varphi))$ and then construct this sentence.

\begin{claim}
If $\varphi=\varphi_0\vee\varphi_1\vee\varphi_2$ satisfies assumptions (a)-(g), then $\mathrm{FOC}(G(n\mid\varphi))$ is totally bounded but spans an infinite-dimensional subspace of $\ell^\infty/c_0$.
\end{claim}
\begin{proof}

Consider any positive integer $n$. The number of graphs that satisfy $\varphi$ equals
\begin{itemize}
\item
$1$, if $n$ is not divisible by $6$;
\item
$1+n!\left(1+\sum\limits_{d \in \mathbb{Z}, d \geq 3, d = \frac{n}{6m^2}} \frac{1}{d!}\right)$, if $n$ is divisible by $6$.
\end{itemize}

Consider a FO sentence $\psi$. Due to (b), (d), and (g), there exist
\begin{itemize}
\item $\beta_{\psi}:= \lim\limits_{n \to \infty}\mathrm{Pr}(G(n\mid\varphi_0)\models \psi)\in\{0,1\}$,
\item $p_{1,\psi}:=\lim\limits_{m \to \infty}\mathrm{Pr}(G(6m\mid\varphi_1)\models \psi)\in[0,1]$,
\item $\beta_{d,\psi}:=\lim\limits_{m \to \infty}\mathrm{Pr}(G(6dm^2\mid\varphi_{2, d})\models \psi)\in\{0,1\}$.
\end{itemize}

Let us define the sequence $p_{\psi}(n)$, $n\in\mathbb{N}$, in the following way: if $ 6 \mid n$, then
$$
p_\psi(n) = 
\left.\left(p_{1, \psi}+\sum\limits_{d \in \mathbb{Z}, d \geq 3, d = \frac{n}{6m^2}} \frac{\beta_{d, \psi}}{d!}\right)\right/ \left(1+\sum\limits_{d \in \mathbb{Z}, d \geq 3, d = \frac{n}{6m^2}} \frac{1}{d!}\right);
$$
otherwise, $p_\psi(n) = \beta_\psi.$ Then, we have
\begin{align}\label{eq:limThm3}
\lim\limits_{n \to \infty}\left|\mathrm{Pr}(G(n\mid\varphi)\models \psi) - p_{\psi}(n)\right| = 0.
\end{align}

Indeed, due to (b), (c), (e), and (f), the number of graphs on $[n]$ that satisfy $\varphi_0 \wedge \psi$ is $\beta_\psi$, for sufficiently large $n$; the number of graphs on $[n=6m]$ that satisfy $\varphi_1 \wedge \psi$ is $p_{1, \psi} n! (1 + o(1))$; the number of graphs on $[n=6dm^2]$ that satisfy $\varphi_{2, d} \wedge \psi$ is $\beta_{d, \psi} \frac{n!}{d!} (1 + o(1))$, and is bounded from above by $\frac{n!}{d!}$. Then, $6 \nmid n$ implies $\mathrm{Pr}(G(n\mid\varphi)\models \psi) = \beta_\psi$ for sufficiently large $n$, and $6 \mid n$ implies
\begin{align}
\mathrm{Pr}(G(n\mid\varphi)\models \psi) & = \frac{p_{1, \psi} + \sum\limits_{d \in \mathbb{Z}, d \geq 3, d = \frac{n}{6m^2}} \beta_{d, \psi} \frac{1}{d!}}{1+\sum\limits_{d \in \mathbb{Z}, d \geq 3, d = \frac{n}{6m^2}} \frac{1}{d!}} + o(1).
\end{align}

For any $d'\geq 3$, let us consider $\psi_{d'}:=\varphi_{2,d'}$. From (a), (c), and (f), observe that
$$
 p_{\psi_{d'}}(n)=\left\{ 
\begin{array}{cc} 
\left(d'!\left(1+\sum\limits_{d \in \mathbb{Z}, d \geq 3, d = \frac{n}{6m^2}} \frac{1}{d!}\right)\right)^{-1}, & \text{if } \frac{n}{6d'} \text{ is a square}; \\
0, & \text{otherwise.}
\end{array} \right.
$$

Let $D$ be the set of all square-free integers $d'\geq 3$. Vectors $\pi\left(\left(p_{\psi_{d'}}(n) \right)_{n \in \mathbb{N}}\right)$, $d' \in D$, are linearly independent since, for each $d' \in D$, the sequence $n = 6 d' m^2$ satisfies $p_{\psi_{d'}}(n) > \frac{1}{d'!e}$, and, for each $d'' \in D$ such that $d'' \neq d'$, and $n = 6 d' m^2$, we have that $p_{\psi_{d''}} (n) = 0$. So, any finite non-trivial linear combination of $\left(p_{\psi_{d''}}(n) \right)_{n \in \mathbb{N}}$ involving $\left(p_{\psi_{d'}}(n) \right)_{n \in \mathbb{N}}$ does not equal $0$. Due to~\eqref{eq:limThm3}, $\mathrm{FOC}(G(n\mid\varphi))$ spans an infinite-dimensional space.

Next, we prove that $\mathrm{FOC}(G(n\mid\varphi))$ is totally bounded. Consider $\varepsilon > 0$. Since all $\beta_{d, \psi}$ are at most $1$, there is a $d_0$ such that, if we restrict the summation in the definition of $p_{\psi}(n)$ to $d\leq d_0$, then we get an $\frac{\varepsilon}{2}$-approximation, i.e. 
$$
\left|p_\psi(n) - \frac{p_{1, \psi}+\sum\limits_{d \in \mathbb{Z}, \, d_0 \geq d \geq 3, \, d = n/(6m^2)} \frac{\beta_{d, \psi}}{d!}}{1+\sum\limits_{d \in \mathbb{Z}, \, d \geq 3, \, d = n/(6m^2)} \frac{1}{d!}}\right| < \frac{\varepsilon}{2}.
$$
Let us construct a finite $\varepsilon$-covering for the set of all sequences $p_{\psi}(n)$ in $\ell^{\infty}$. Let $N$ be a positive integer such that $\frac{1}{N} < \frac{\varepsilon}{2}$. The desired $\varepsilon$-covering is the set of all sequences $v_{k,\mathbf{b}}(n)$ indexed by $k\in[N]$ and $\mathbf{b}=(b,b_3,b_4,\ldots,b_{d_0})\in\{0,1\}^{d_0-1}$ and defined in the following way: if $6 \mid n$, then
$$
v_{k, \mathbf{b}}(n)=\left.\left(\frac{k}{N}+\sum\limits_{d_0 \geq d \geq 3, \, d = n/(6m^2)} \frac{b_d}{d!}\right)\right/ \left(1+\sum\limits_{d \geq 3, \, d = n/(6m^2)} \frac{1}{d!}\right);
$$
otherwise, $v_{k, \mathbf{b}}(n)=b$. Thus, this family of sequences is indeed the desired $\varepsilon$-covering. Due to~\eqref{eq:limThm3}, $\mathrm{FOC}(G(n\mid\varphi))$ is totally bounded, completing the proof.
\end{proof}

To finish the proof of (iii), we construct the FO sentences $\varphi_1$ in Section~\ref{varphi1} and $\varphi_2$ in Section~\ref{varphi2}. 

\subsection{Definition of $\varphi_1$ and verification of its properties}\label{varphi1}

Let $\Delta_m$ be the set of all digraphs on $[m]$ with all in-degrees and out-degrees equal $1$. In particular, a loop contributes $1$ both to the in-degree and the out-degree of the respective vertex. Such digraphs are disjoint unions of oriented cycles and there are exactly $m!$ digraphs in $\Delta_m$. For each class of isomorphism $D$ of digraphs on $m$ vertices we construct a class of isomorphism $G(D)$ of simple graphs on $6m$ vertices in the following way. We replace every edge $(u, v) \in D$ with a graph on the set of vertices $\{u, v, w_i(u, v), i \in [5]\}$, where all vertices $w_i(u, v)$ are different, comprising edges $\{u, w_{1}(u, v)\}$, $\{w_{1}(u, v), w_{3}(u, v)\}$, $\{w_{3}(u, v), v\}$, $\{w_{1}(u, v), w_{2}(u, v)\}$,  $\{w_{3}(u, v), w_{4}(u, v)\}$, and $\{w_{4}(u, v), w_{5}(u, v)\}$.

Let us show that the property $\bigcup\limits_m\Gamma_m$ is FO, where  $\Gamma_m = \{G(D), D \in \Delta_m\}$, and define $\varphi_1$ as a sentence describing this property. Let $Type_{i_1, \dots, i_j}(x)$ be a FO formula saying that $x$ has degree $j$ and adjacent to vertices of degrees $i_1, \dots, i_j$. The sentence $\varphi_1$ is conjunction of the following sentences:
\begin{itemize}
\item
$\forall x\,\, Type_{2}(x) \vee Type_{3}(x) \vee Type_{1, 3}(x) \vee Type_{3, 3}(x) \vee Type_{1, 2, 3}(x) \vee Type_{2, 2, 3}(x);$
\item
$
\forall x\,\, Type_{1, 2, 3}(x) \Rightarrow \left(\exists y\,\, Type_{2, 2, 3}(y) \wedge (x \sim y)\right);
$
\item
$
\forall x\,\, Type_{2, 2, 3}(x) \Rightarrow \left(\exists y\,\, Type_{1, 2, 3}(y) \wedge (x \sim y)\right);
$
\item
$\forall x\,\, (Type_{2, 2, 3}(x) \vee Type_{1, 2, 3}(x)) \Rightarrow\left(\exists! y\,\, Type_{3, 3}(y) \wedge (x \sim y)\right);$
\item
$\forall x\,\, Type_{3, 3}(x) \Rightarrow\left(\exists y \exists z\,\, Type_{1, 2, 3}(y) \wedge Type_{2, 2, 3}(z) \wedge (x \sim y) \wedge (x \sim z)\right).$
\end{itemize}

Let us prove that $\varphi_1$ indeed describes $\bigcup\limits_m\Gamma_m$. Obviously, any graph having the property $\Gamma_m$ satisfies $\varphi_1$. Let $G \models \varphi_1$. We have to show that $G \in \bigcup\limits_m\Gamma_m$. Let us first show that all types listed in the first clause of of $\varphi_1$ are presented in $G$. All vertices of $G$ have degree at most $3$. Let us show that there is a vertex of degree exactly $3$. Suppose the opposite, then there is no vertex of degree $2$ because $\varphi_1$ forces each vertex of degree $2$ to be adjacent to a vertex of degree $3$. Therefore, all vertices have degree $1$ that leads to contradiction because due to the definition of $\varphi_1$ vertices of degree $1$ cannot be adjacent to other vertices of degree $1$. Next, note that the presence of a degree $3$ vertex of any of the two types implies the presence of a degree $3$ vertex of the second type. Vertex that satisfies $Type_{1, 2, 3}$ has neighbours that satisfy $Type_{3}$ and $Type_{3, 3}$. Vertex that satisfies $Type_{2, 2, 3}$ has exactly one neighbour that satisfies $Type_{3, 3}$. Hence, another neighbour of degree $2$ satisfies $Type_{1, 3}$. Finally, a vertex that satisfies $Type_{1, 3}$ is adjacent to a vertex with degree $1$ that satisfies $Type_2$, i.e. all types allowed by $\varphi_1$ are indeed presented.

Let us reconstruct $D \in \Delta_m$ such that $G \cong G(D)$. It would imply that $G \in \Gamma_m$. We set the vertices of $D$ to be the vertices of $G$ that satisfy $Type_{3, 3}$. We add an edge $(u, v)$ between (not necessarily distinct) vertices of $D$ if and only if there is a path $uxyv$ in $G$ such that $Type_{1, 2, 3}(x) \wedge Type_{2, 2, 3}(y)$ holds. Due to $\varphi_1$, each vertex $u$ that satisfies $Type_{3, 3}(u)$ is adjacent to a unique vertex $x$ that satisfies $Type_{1, 2, 3}(x)$; the vertex $x$ is adjacent to a unique vertex $y$ that satisfies $Type_{2, 2, 3}(y)$; the vertex $y$ is adjacent to a unique vertex $v$ (not necessarily $v \neq u$) that satisfies $Type_{3, 3}(y)$. So, for each vertex of $D$ its out-degree equals $1$. Similarly, the in-degree of each vertex in $D$ equals $1$. Therefore, the reconstructed digraph is in $\Delta_m$. It remains to show that the graph $G$ is isomorphic to $G(D)$, and thus belongs to $\Gamma_m$. Indeed, vertices that satisfy $Type_{3, 3}$ are the vertices of $D$. By the definition of edges in $D$, we have the respective vertices $w_{1}(u, v)$ and $w_{3}(u, v)$, and the conditions on types of these vertices guarantee the presence of properly connected vertices $w_{2}(u, v), w_{4}(u, v), w_{5}(u, v)$ in $G$. All vertices $w_i(u, v)$ are different because for different $i$ they have different types, and $u, v$ can be restored as the two (or one $u = v$) nearest vertices that satisfy $Type_{3, 3}$. There are no other vertices due to the first clause in the definition of $\varphi_1$.

 Let us now check the condition (c). For each $m$, the set $\Gamma_m$ consists only of graphs of size $6m$. Then, there are no graphs that satisfy $\varphi_1$ of size not divisible by $6$. Let $D_1, \dots, D_t$ be all classes of isomorphism of digraphs in $\Delta_m$. Then, we have $1 = \sum_{D_i} \frac{1}{|\mathrm{Aut}(D_i)|}$. Also, $|\mathrm{Aut}(D_i)| = |\mathrm{Aut}(G(D_i))|$. Then the number of graphs on $[6m]$ that satisfy $\varphi_1$ equals
$$
\sum_{D_i} \frac{(6m)!}{|\mathrm{Aut}(G(D_i))|} = \sum_{D_i} \frac{(6m)!}{|\mathrm{Aut}(D_i)|} = (6m)!\sum_{D_i} \frac{1}{|\mathrm{Aut}(D_i)|} = (6m)!.
$$

Let us finally verify the condition (d). Fix a FO sentence $\psi$ of quantifier depth $d$. By Theorem~\ref{hanf}, for $D, D' \in \Gamma_m$ with equal numbers of components of all sizes at most $6 \cdot 2^d$, $\psi$ does not distinguish between $D$ and $D'$. For every $i \leq 2^d$, let $X_i$ be the number of components in $G(n \mid \varphi_1)$ of size $6i$; let $Y_i$ be the number of components of size $i$ in a uniformly random digraph from $\Delta_m$. Note that $(X_1, \dots, X_{2^d}) \stackrel{d}{=} (Y_1, \dots, Y_{2^d})$, and that $(Y_1, \dots, Y_{2^d}) \stackrel{d}{\to} (\xi_1, \dots, \xi_{2^d})$ due to Theorem~\ref{cyclesDistr}, where $\xi_i$ are independent $\mathrm{Pois}(\frac{1}{i})$ random variables. Then, 
\begin{align*}
\mathrm{Pr}(G(6m\mid\varphi_1)\models \psi) &= \sum \mathrm{Pr}(X_1 = x_1, \dots, X_{2^d} = x_{2^d})\\
&=(1+o(1))\sum \mathrm{Pr}(\xi_1 = x_1, \dots, \xi_{2^d} = x_{2^d})
\end{align*}
completing the proof of (d).

\subsection{Definition of $\varphi_2$ and verification of its properties}\label{varphi2}

Next, we construct the FO sentence $\varphi_2$ and sentences $\varphi_{2,d}$. Let us first, for every integer $m \geq 2$, define an auxiliary graph $L_m$. This graph consists of vertices $u_i$, for $1 \leq i \leq 6m$; $v_{i, j}$, for $j \in [m - 1]$, $6j+1 \leq i \leq 6m$; and $w_{i, j}$, for $j \in [m - 1]$, $6j+1 \leq i \leq 6m$; and edges
\begin{itemize}
\item
$\{u_{i-1}, u_i\}$, for $2 \leq i \leq 6m$;
\item
$\{u_i, v_{i, j}\}, \{v_{i, j}, w_{i, j}\}$, for $j \in [m - 1]$, $6j + 1 \leq i \leq 6m$;
\item
$\{v_{i-1,j}, v_{i, j}\}$, for $j \in [m - 1]$, $6j + 2 \leq i \leq 6m$.
\end{itemize}

The graph $L_m$ consists of $6m + 6\frac{m(m-1)}{2} + 6\frac{m(m-1)}{2} = 6m^2$ vertices and has no nontrivial automorphisms.  Let us show that there is a FO sentence $\varphi_L$ that expresses the property of being isomorphic to $L_m$ for some $m \geq 2$. Let $\varphi_L = \mathrm{Types} \wedge \mathrm{UDeg} \wedge \mathrm{VDeg} \wedge \mathrm{VUEdges} \wedge \mathrm{VUSquare} \wedge \mathrm{UVPattern}_1 \wedge \mathrm{UVPattern}_2 \wedge \mathrm{UVPattern}_3$, where the clauses in the conjunction are defined as follows.

\begin{itemize}
\item
$\mathrm{Deg}_{i}(x)$ ($\mathrm{Deg}_{\geq i}(x)$) is a FO formula that expresses the property of the vertex $x$ to have degree $i$ (at least $i$). 
\item
$W(x) = \mathrm{Deg}_1(x) \wedge \left(\exists y \,\, (x \sim y) \wedge \mathrm{Deg}_{\geq 3}(y)\right)$.
\item
$V(x) = \exists! y \,\, (x \sim y) \wedge W(y)$.
\item
$U(x) = \neg W(x) \wedge \neg \exists y \,\, (x \sim y) \wedge W(y)$.
\item
$\mathrm{Types} = \forall x \,\, W(x) \vee V(x) \vee U(x)$.
\item
$\mathrm{UDeg}$ is a FO sentence saying that the induced subgraph on $\{x:U(x)\}$ consists of vertices with degrees $1$ or $2$, and exactly two of them have degree $1$.
\item
$\mathrm{VDeg}$ is a FO sentence saying that the induced subgraph on $\{x:V(x)\}$ consists of vertices of degrees $1$ or $2$.
\item
$\mathrm{VUEdges} = \forall x \,\, \left(V(x) \Rightarrow \left(\exists! y \,\, U(y) \wedge (x \sim y)\right)\right)$. This sentence defines a mapping $\mathrm{UVMap}: \{x: V(x)\} \to \{x: U(x)\}$.
\item
$\mathrm{VUSquare}$ is a FO sentence saying that, for any pair of adjacent vertices $x,x'$ satsifying $V(x)\wedge V(x')$, their images $\mathrm{UVMap}(x)$ and $\mathrm{UVMap}(x')$ are adjacent as well.
\end{itemize}


It remains to define $\mathrm{UVPattern}_1$, $\mathrm{UVPattern}_2$ and $\mathrm{UVPattern}_3$. For FO formulae $\psi_1, \psi_2$ with a free variable $x$, let 
\begin{align*}
\mathrm{Matching}_x[\psi_1, \psi_2] = & \forall x \,\, \neg(\psi_1(x) \wedge \psi_2(x)) \wedge \\
& \left(\psi_1(x) \Rightarrow \exists! \,\, x' (x \sim x') \wedge \psi_2(x') \right) \wedge \\
& \left(\psi_2(x) \Rightarrow \exists! \,\, x' (x \sim x') \wedge \psi_1(x') \right)
\end{align*}
be the formula saying that the set $A_1$ of $x$ satisfying $\psi_1(x)$ and the set $A_2$ of $x$ satisfying $\psi_2(x)$ are disjoint and edges between them form a perfect matching.

\begin{itemize}
\item
$\mathrm{VMatching}_x(y, y') = \mathrm{Matching}_x[V(x) \wedge (x \sim y), V(x) \wedge (x \sim y')]$.
\item 
$\mathrm{VAlmostMatching}(y, y')$ is the FO sentence 
\begin{multline*}
\exists z' \,\, V(z') \wedge (y' \sim z')
 \wedge \left( \forall z \,\, (V(z) \wedge (y \sim z)) \Rightarrow \neg(z \sim z') \right) \wedge\\
 \wedge\mathrm{Matching}_x[V(x) \wedge (x \sim y), V(x) \wedge (x \sim y') \wedge (x \neq z')].
\end{multline*}
In other words, the induced bipartite graph between $A=\{x: V(x) \wedge (x \sim y)\}$ and $B=\{x: V(x) \wedge (x \sim y')\}$ is a disjoint union of a matching and an isolated vertex $z' \in B$.
\item
$\mathrm{UVPattern}_1$ is a FO sentence saying that, for each $y_0, y_1, \ldots$, $y_6$ satisfying 
\begin{align}\label{eq:conditionThm3}
\bigwedge_{0 \leq i \leq 6} U(y_i) \wedge \bigwedge_{0 \leq i \leq 5} (y_i \sim y_{i+1}) \wedge \bigwedge_{0 \leq i \leq 4} (y_i \neq y_{i+2}),
\end{align}
there are five $i \in [6]$ such that $\mathrm{VMatching}(y_{i - 1}, y_i)$ holds and for the single remaining $i \in [6]$, at least one of $\mathrm{VAlmostMatching}(y_{i - 1}, y_i)$ and $\mathrm{VAlmostMatching}(y_i, y_{i - 1})$ is satisfied.
\end{itemize}

In order to define $\mathrm{UVPattern}_2$ and $\mathrm{UVPattern}_3$, we need auxiliary formulae $\mathrm{UStart}$ and $\mathrm{UEnd}$:

\begin{itemize}
\item
$\mathrm{UStart}(x) = U(x) \wedge Deg_1(x)$.
\item 
$\mathrm{UEnd}(x) = U(x) \wedge Deg_{\geq 2}(x) \wedge \left(\exists! y \,\, U(y) \wedge (x \sim y)\right)$.
\item 
Let $\mathrm{UVPattern}_2$ be a FO sentence saying that, for each $y_0$, $y_1, \dots, y_6$ satisfying~\eqref{eq:conditionThm3}, 
\begin{multline*}
\left(\mathrm{UStart}(y_0) \Rightarrow \mathrm{VAlmostMatching}(y_5, y_6) \right) \wedge\\
\wedge \left(\mathrm{UEnd}(y_6) \Rightarrow \mathrm{VAlmostMatching}(y_0, y_1)\right),
\end{multline*} 
is satisfied, and there are vertices $y$ and $y'$ satisfying $\mathrm{UStart}(y)$ and $\mathrm{UEnd}(y')$.
\item 
Let $\mathrm{UVPattern}_3$ be a FO sentence saying that, for each $y_0$, $y_1, \dots, y_7$ satisfying
$$
\quad\bigwedge_{0 \leq i \leq 7} U(y_i) \wedge \bigwedge_{0 \leq i \leq 6} (y_i \sim y_{i+1}) \wedge \bigwedge_{0 \leq i \leq 5} (y_i \neq y_{i+2}),
$$
the formula
$$
\quad\,\,\mathrm{VAlmostMatching}(y_0, y_1) \Rightarrow \mathrm{VAlmostMatching}(y_6, y_7)
$$
is satisfied.
\end{itemize}

Let us briefly verify that $\varphi_L$ expresses the desired property of being isomorphic to some $L_m$. We skip a direct routine check that $L_m$ satisfies all clauses in the definition of $\varphi_L$. Let $G \models \varphi_L$ and let us prove that $G \cong L_m$ for a certain $m$. By $\mathrm{UDeg}$, we know that the subgraph $G_U$ induced on vertices $x$ that satisfy $U(x)$ is a disjoint union of a single path and cycles. 

Suppose that there is a cycle $y_0\ldots y_{s-1}y_0$ in $G_U$ for some $s\geq 3$. Then, without loss of generality, the sentence $\mathrm{UVPattern}_1$ implies that $\mathrm{VAlmostMatching}(y_0, y_1)$ holds. For the sake of convenience, set $y_k := y_{k - s}$ for $k \geq s$. For every $i\in\mathbb{Z}_{\geq 0}$, denote by $n_i$ the number of neighbours $x$ of $y_i$ that satisfy $V(x)$. Then, $\mathrm{VMatching}(y_{i}, y_{i+1})$ implies $n_{i+1} = n_i$, and $\mathrm{VAlmostMatching}(y_{i}, y_{i+1})$ implies $n_{i+1}=n_i + 1$. Due to $\mathrm{UVPattern}_1 \wedge \mathrm{UVPattern}_3$, for each $i$ divisible by $6$, $n_{i+1} = n_i + 1$, and $n_{i+1} = n_i$ for other. Therefore, the sequence $n_0, n_6, n_{12}, \dots, n_{6s}$ is strictly increasing. It leads to contradiction because $y_0=y_{6s}$ and so $n_0 = n_{6s}$. Therefore, $G_U$ is a path. 

By $\mathrm{VDeg}$, we have that there are no vertices of degree more than $2$ in the subgraph $G_V$ induced on $\{x: V(x)\}$. Hence, this graph is a disjoint union of paths and cycles. By $\mathrm{VUEdges} \wedge \mathrm{VUSquare}$, we have the mapping $\mathrm{UVMap}: G_V \to G_U$ that sends adjacent vertices to adjacent. Also, by $\mathrm{UVPattern}_1$, $\mathrm{UVMap}$ sends two different vertices with a common neighbour in $G_V$ to two different vertices in $G_U$. Indeed, if $x_1,x_2$ are neighbours of $x$ in $G_V$ and $y$ is a common neighbour of $x_1,x_2$ from $G_U$, then there is a vertex $y'\in G_U$ which is a common neighbour of $x$ and $y$. This contradicts $\mathrm{UVPattern}_1$ since, if the latter holds, then either $\mathrm{VMatching}(y,y')$, or $\mathrm{VAlmostMatching}(y,y')$, or $\mathrm{VAlmostMatching}(y',y)$. All these predicates do not allow vertices of degree more than 1 in the bipartite graph between the neighbourhoods of $y$ and $y'$ in $G_V$. Since the graph induced on $G_U$ is a path, there are no cycles in $G_V$, i.e. $G_V$ is a disjoint union of paths, where each path is mapped by $\mathrm{UVMap}$ to a subpath in $G_U$. By $\mathrm{UVPattern}_2 \wedge \mathrm{UVPattern}_3$, we have that these paths in $G_V$ behave as in $L_m$ in the following sense: the induced subgraph on $V(G_U \cup G_V)$ is isomorphic to the induced subgraph on $u_i$ and $v_{i, j}$ in $L_m$; vertices of $G_U$ correspond to vertices $u_i$ and vertices of $G_V$ correspond to vertices $v_{i, j}$. By $\mathrm{Types}$, remaining vertices have degree one and adjacent to vertices in $G_V$. Since for each $x \in G_V$ there is exactly one neighbor $y$ that satisfies $W(x)$, we have the correspondence between such vertices $y$ and vertices $w_{i, j}$ in $L_m$ that completes the isomorphism between $G$ and~$L_m$.

Now, we fix the presented above sentence $\varphi_L$ and consider the family of graphs $K_d \Box L_m$, for $d \geq 3$ and $m \geq 2$, where $G \Box H$ is the Cartesian product of graphs, see~Definition~\ref{defCartesian}. We construct a FO sentence $\varphi_2$ that expresses the property of being isomorphic to $K_d \Box L_m$ for some $d \geq 3$ and $m \geq 2$ and, for every $d \geq 3$, we construct a FO sentence $\varphi_{2, d}$ that expresses the property of being isomorphic to $K_d \Box L_m$ for some $m \geq 2$. We first set 
$$
\varphi_2 = \mathrm{TEquiv} \wedge \mathrm{TGraph} \wedge \varphi_{TL} \wedge \mathrm{TCommute},
$$
where the clauses are defined in the following way. Let 
$$
\mathrm{Triangle}(x, y) = (x = y) \vee \left(\exists z \,\, (x \sim y) \wedge (y \sim z) \wedge (z \sim x)\right)
$$
express the property of distinct $x$, $y$ to belong to a triangle. For FO formulae $\psi_1, \psi_2$ with a free variable $x$, let 
\begin{align*}
\mathrm{Matching}_x[\psi_1, \psi_2] = & \forall x \,\, \neg(\psi_1(x) \wedge \psi_2(x)) \wedge \\
& \left(\psi_1(x) \Rightarrow \exists! \,\, x' (x \sim x') \wedge \psi_2(x') \right) \wedge \\
& \left(\psi_2(x) \Rightarrow \exists! \,\, x' (x \sim x') \wedge \psi_1(x') \right)
\end{align*}
say that sets $\{x\mid\psi_1(x)\}$ and $\{x\mid\psi_2(x)\}$ are disjoint and edges between them form a perfect matching.

\begin{itemize}
\item $\mathrm{TEquiv}$ is a FO sentence saying that $\mathrm{Triangle}(x, y)$ is an equivalence relation and each equivalence class has size at least $3$ (or, in other words, every vertex belongs to a triangle).
\item $\mathrm{TGraph}$ is a FO sentence saying that for each pair $y, y'$ such that $\neg \mathrm{Triangle}(y, y')$ holds, either there are no edges between their $\mathrm{Triangle}$-equivalence classes, or  
$$
{TEdge}(y, y') := \mathrm{Matching}_x[\mathrm{Triangle}(x, y), \mathrm{Triangle}(x, y')]
$$ 
holds.
\item $\varphi_{TL}$ is the sentence $\varphi_L$ with all predicates $x = y$ replaced by $\mathrm{Triangle}(x, y)$, and all $x \sim y$ replaced by $\mathrm{TEdge}(x, y)$.
\item $\mathrm{TCommute}$ is a FO sentence saying that, for each four vertices $x, x', y, y'$ such that $\mathrm{TEdge}(x, x') \wedge \mathrm{TEdge}(y, y') \wedge\mathrm{TEdge}(x, y) \wedge \mathrm{TEdge}(x', y')$, the subgraph induced on their $\mathrm{Triangle}$-equivalence classes is a disjoint union of cycles of length $4$ with one vertex from each class.
\end{itemize}
Finally, $\varphi_{2, d}$ is the conjunction of $\varphi_2$ and a FO sentence saying that there is a clique on $d$ vertices but no cliques on $d+1$ vertices. So, we immediately have (e).

Let us show that $\varphi_2$ and $\varphi_{2, d}$ express the proper sets of graphs. As usual, we omit the straightforward verification that $K_d\Box L_m\models\varphi_{2,d}$. Assume $G \models \varphi_2$. By $\mathrm{TEquiv} \wedge \mathrm{TGraph}$, the set of vertices of $G$ is partitioned into $\mathrm{Triangle}$-equivalence classes; edges between vertices of two different equivalence classes $B,C$ in $G$ appear if and only if representatives $x\in B$, $y\in C$ satisfy $\mathrm{TEdge}(x,y)$. Let us consider an auxiliary graph $\tilde G$ whose vertices are the $\mathrm{Triangle}$-equivalence classes of $G$; two vertices $B,C$ are adjacent in $\tilde G$ if and only if there are edges between them in $G$ (and, in this case, these edges in $G$ between $B$ and $C$ compose a perfect matching). Due to $\varphi_{TL}$, $\tilde G$ is isomorphic to some $L_m$. Since $L_m$ is connected, all $\mathrm{Triangle}$-equivalence classes in $G$ have the same size $d$. Moreover, for a certain $d\geq 3$, these classes induce cliques $K_d$ because, for distinct vertices $x,y$, $\mathrm{Triangle}(x,y)$ implies $x \sim y$. So, $G \models \varphi_{2, d}$.

For two different $\mathrm{Triangle}$-equivalence classes $B$ and $C$, the perfect matching between them defines two bijections $f_{BC}:B\to C$ and $f_{CB}:C\to B$ in the natural way: for $x\in B$, $y\in C$, the adjacency $x\sim y$ implies $f_{BC}(x)=y$ and $f_{CB}(y)=x$. Note that $f_{BC} \circ f_{CB} = id_C$ and $f_{CB} \circ f_{BC} = id_B$. By $\mathrm{TCommute}$ we have that, for each cycle $BB'C'C$ of length $4$ in $\tilde G$, $f_{CC'} \circ f_{BC} = f_{B'C'} \circ f_{BB'}$. Note that the 2-dimensional CW-complex obtained by ``filling'' all 4-cycles of $L_m$ is simply connected. Or, in algebraic language, for a certain spanning subtree $\hat L_m\subset L_m$, the group presented by $\langle f_{BC},\,BC\in E(L_m)\mid f_{BC},\, BC\in E(\hat L_m),\, f_{BC}f_{CB}, \, BC\in E(L_m),\, f_{B'B}f_{C'B'}f_{CC'}f_{BC},\,BB'CC'\text{ is 4-cycle in }L_m\rangle$ is trivial. So, for every two walks between $B,C\in V(\tilde G)$, compositions of respective bijections along the walks are equal. This fact legitimises the definition of $f_{BC}$ for any pair of $B,C\in V(\tilde G)$: let $BP_0\ldots P_k C$ be a path in $\tilde G$, then $f_{BC}:=f_{P_k C} \circ f_{P_{k-1} P_k} \circ \cdots f_{P_0 P_1} \circ f_{B P_0}$. Next, fix a $\mathrm{Triangle}$-equivalence class $B = \{b_1, \dots, b_d\}$. For each $b_i$, let $S_i = \{f_{BC}(b_i), C \in V(\tilde G)\}$. By the definition of $f_{BC}$, there are no edges $\{f_{BC}(b_i), f_{BC'}(b_j)\}$ for $i \neq j$ and $C \neq C'$, i.e. each $\tilde G$-edge $\{C, C'\}$ is presented by $G$-edges $\{f_{BC}(b_i), f_{BC'}(b_i)\}$, $1 \leq i \leq d$. Therefore, the induced subgraph on $S_i$ is isomorphic to $L_m$ and, consequently, the graph $G$ itself is isomorphic to $K_d \Box L_m$.  So, $\varphi_2$ and $\varphi_{2,d}$ express desired properties of finite graphs.

The number of vertices in $K_d \Box L_m$ is $6dm^2$. For each $n \neq 6 d m^2$, there are no graphs that satisfy $\varphi_{2, d}$. For $n = 6 d m^2$, there is exactly one graph under isomorphism that satisfy $\varphi_{2, d}$. The graph $K_d \Box L_m$ has exactly $d!$ automorphisms. 
 We conclude that the number of graphs that satisfy $\varphi_{2, d}$ equals $\frac{n!}{d!}$, completing the proof of (f).

To finish the proof, it remains to show (g). Due to Ehrenfeucht's theorem~\cite{Eh, Fr}, we know that there is a FO sentence of quantifier depth $k$ that distinguishes between two graphs $G$ and $H$ if and only if Spoiler has a winning strategy in $k$ moves in the Ehrenfeucht--Fra\"{\i}ss\'{e} game on $G$ and $H$. Suppose, there is a FO sentence $\psi$ of the quantifier depth $k$ such that $\mathrm{Pr}(G(6dm^2|\varphi_{2,d}) \models \psi)$ does not converge to either $0$ or $1$. This means that there are infinitely many $m$ such that $K_d \Box L_m \models \psi$ and infinitely many $m$ such that $K_d \Box L_m \not\models \psi$. Therefore, it is enough to show that, for every fixed positive integer $k$ and every fixed integer $d\geq 3$, in the Ehrenfeucht--Fra\"{\i}ss\'{e} game Spoiler has no winning strategy in $k$ moves, on graphs $K_d \Box L_m$ and $K_d \Box L_{m'}$, for large enough $m$ and $m'$. Due to Claim~\ref{Cartesian}, we may get rid of the Cartesian product and prove the same fact for graphs $L_m$ and $L_{m'}$. 

For $a<b\leq m$, let us consider the following induced subgraphs of $L_m$:
\begin{itemize}
\item $L_{(a,b]}=L_m\left[\{u_i, v_{i,j}, w_{i,j} \mid 6a+1 \leq i \leq 6b, 6j+1 \leq i\}\right]$; 

\item $Z_{(a,b]}=L_m\left[\{u_i, v_{i,j}, w_{i,j} \mid 6a+1 \leq i \leq 6b, j \leq a\}\right]$;

\item $\tilde Z_{(a,b]}=L_m\left[\{u_i, v_{i,j} \mid 6a+1 \leq i \leq 6b, j \leq a\}\right]$. 
\end{itemize}

Consider mappings $f_{a, b}$ of paths $X_{b - a}=u_1 \dots u_{6(b-a)}$ to $Z_{(a, b]}$ such that $f_{a, b}(u_i) = u_{6a + i}$. Note that $L_{(a, b]} \cong Z_{(a, b]} \cup_{f_{a, b}} L_{(0, b-a]}$. Let $\mathrm{Star}_r$ be a star graph with $r$ leaves. So, $\tilde Z_{(a, b]} \cong X_{b - a} \Box \mathrm{Star}_a$. For two numbers $a, a' > k$, Duplicator has a winning strategy in the Ehrenfeucht--Fra\"{\i}ss\'{e} game with $k$ moves on graphs $\mathrm{Star}_a$ and $\mathrm{Star}_{a'}$, this winning strategy preserves central vertices of stars. By Claim~\ref{Cartesian}, Duplicator has a winning strategy in the Ehrenfeucht--Fra\"{\i}ss\'{e} game in $k$ round on graphs $\tilde Z_{(a, b]}$ and $\tilde Z_{(a', b']}$, where $b-a = b'-a'$. Next, note that $Z_{a, b}$ is obtained from $\tilde Z_{(a, b]}$ by attaching a leaf to each vertex $v_{i, j}$. Since the winning strategy of Duplicator on stars is leaves-preserving, we can apply Claim~\ref{leaves} with $I$ distinguishing between ``leaves'' and ``non-leaves'' and get a winning strategy of Duplicator in the game on $Z_{(a,b]}$ and $Z_{(a',b']}$. Finally, due to Claim~\ref{glue}, Duplicator has a winning strategy in the game on $L_{(a,b]}$ and $L_{(a',b']}$. This makes it possible to apply the well-known Duplicator's strategy on two long paths in the game on $L_m,L_{m'}$ since $L_m$ and $L_{m'}$ can be represented as unions of three segments $L_{(a,b]}$ (the first one, for $a=0$, an intermediate, and the last one, for $b\in\{m,m'\}$)  such that the respective segments are elementary equivalent. For completeness, let us recall the strategy.

Let $p_m: V(L_m) \to [m]$ maps each vertex $u_i, v_{i, j}, w_{i, j}$ to the number $\lceil \frac{i}{6} \rceil$. Similarly define $p_{m'}$. Suppose that $m, m' > 3^{k+1}$ and in first $t-1$ rounds vertices $x_1, \dots, x_{t-1} \in L_m$ and $x'_1, \dots, x'_{t-1} \in L_{m'}$ are chosen. Without loss of generality, Spoiler chooses a vertex $x_t \in L_m$ in round $t$. 
\begin{itemize}
\item
If $p_m(x_t) \leq 3^{k - t} + 1$, Duplicator chooses the same vertex $x_t$ in $L_{m'}$. 
\item
If $m - p_m(x_t) \leq 3^{k - t}$, Duplicator chooses the next vertex according to the strategy in the game on $L_{(m - 3^{k - t} - 1, m]}$ and $L_{(m' - 3^{k - t} - 1, m']}$. 
\item
If there is $x_s$ with $s < t$ such that $|p_m(x_t) - p_m(x_s)| \leq 3^{k - t}$, Duplicator chooses the next vertex according to the strategy in the game on $L_{(p_m(x_s) - 3^{k - t} - 1, p_m(x_s) + 3^{k - t}]}$ and $L_{(p_{m'}(x'_s) - 3^{k - t} - 1, p_{m'}(x'_s) + 3^{k - t}]}$. 
\item
Otherwise, Duplicator chooses $3^{k - t} + 1 < a < m - 3^{k - t}$ such that $|a - p_{m'}(x'_s)| > 3^{k - t}$ for all $s < t$, and then chooses a vertex according to the strategy in the game on $L_{(p_m(x_s) - 3^{k - t} - 1, p_m(x_s) + 3^{k - t}]}$ and $L_{(a - 3^{k - t} - 1, a + 3^{k - t}]}$. It is possible to choose such an $a$ since $m'$ is large enough. 
\end{itemize}

A straightforward inductive argument implies that, for every $t\leq k$, as soon as $t$ rounds are played on $L_m,L_{m'}$, Duplicator has a winning strategy on pairs of graphs $\left(L_{(0, 3^{k - t} + 1]}, L_{(0, 3^{k - t} + 1]}\right)$,  $\left(L_{(m - 3^{k - t} - 1, m]}, L_{(m' - 3^{k - t} - 1, m']}\right)$, $\left(L_{(p_m(x_s) - 3^{k - t} - 1, p_m(x_s) + 3^{k - t}]}\right.,$ \linebreak $\left.L_{(p_{m'}(x'_s) - 3^{k - t} - 1, p_{m'}(x'_s) + 3^{k - t}]}\right)$ for all $s\leq t$. It immediately implies that Duplicator wins the game on $L_m,L'_m$.

Thus, (g) follows, completing the proof of part (iii) of Theorem~\ref{TH:AXFOC}. 

\section*{Acknowledgements}

The authors thank Michael Benedikt for helpful discussions and for finding some of the references to existing literature. The work of the first author is supported by RSF grant 22-11-00131.

\appendix

\section{Proof of Lemma~\ref{lm:appendix_Bernoulli}}
\label{AppendixE}

Let $A_n \subset \{0, 1\}^{s_n}$. Then,
$$
\mathrm{Pr}((\xi_1,\ldots,\xi_{s_n}) \in A_n) 
= \sum_{(a_1, \dots, a_{s_n}) \in A_n} p^{a_1 + \cdots + a_{s_n}} (1 - p)^{s_n - a_1 - \cdots - a_{s_n}},
$$
where $a_i \in \{0, 1\}$, and the same equality with $p$ replaced by $q$ holds for $(\eta_1,\ldots,\eta_{s_n})$. Then, letting 
$$
f_{a_1,\ldots,a_n}(p)=p^{a_1 + \cdots + a_{s_n}} (1 - p)^{s_n - a_1 - \cdots - a_{s_n}},
$$
we get
\begin{align} 
|\mathrm{Pr}((\xi_1,\ldots,\xi_{s_n})  \in A_n)  - & \mathrm{Pr}((\eta_1,\ldots,\eta_{s_n}) \in A_n)|\notag  \\
& = \left|\sum_{(a_1, \dots, a_{s_n}) \in A_n} (f_{a_1,\ldots,a_n}(p) - f_{a_1,\ldots,a_n}(q)) \right|\notag  \\
& \leq \sum_{(a_1, \dots, a_{s_n}) \in A_n} \left| f_{a_1,\ldots,a_n}(p) - f_{a_1,\ldots,a_n}(q) \right|\notag \\
& \leq \sum_{k = 0}^{s_n} \binom{s_n}{k} \left| p^k (1 - p)^{s_n - k} - q^k (1 - q)^{s_n - k} \right|. \label{eq:1Lemma3}
\end{align}
We shall prove that this sum is at most $2 C' \sqrt{\frac{s_n}{\pi (p_{\min} - \frac{1}{s_n})}} |p - q|
$, for some constant $C'$ and $p_{\min} \geq \frac{3}{2s_n}$. Without a loss of generality, we suppose that $p > q$ and $q \leq \frac{1}{2}$. The term $p^k (1 - p)^{s_n - k} - q^k (1 - q)^{s_n - k}$ is positive if and only if $\left(\frac{p(1 - q)}{q(1 - p)}\right)^k > \left(\frac{1 - q}{1 - p}\right)^{s_n}$, i.e. 
$$
k > k_0 := \left\lfloor s_n \frac{\ln (1 - q) - \ln (1 - p)}{\ln p + \ln (1 - q) - \ln q - \ln (1 - p)} \right\rfloor.
$$
Let us prove that
\begin{equation}
\lfloor qs_n\rfloor \leq k_0<ps_n.
\label{eq:k_0Lemma3}
\end{equation}
Consider the function $g(x) = x^k (1 - x)^{s_n - k}$. The derivative of this function is 
\begin{align*}
g'(x)  = k x^{k-1} (1 - x)^{s_n - k} - (s_n - k) x^k (1 - x)^{s_n - k - 1} 
 = (k - s_n x)x^{k-1}(1 - x)^{s_n - k - 1}.
\end{align*}
For $p \leq \frac{k}{s_n}$, $g'(x)$ is positive on the interval $(q, p)$. Therefore, $g(p) - g(q) > 0$. This means that $p s_n > k_0$. Similarly, for $q \geq \frac{k}{s_n}$, $g(p) - g(q) < 0$, and then $\lfloor q s_n \rfloor \leq k_0$, completing the proof of~\eqref{eq:k_0Lemma3}. We then get
\begin{multline}
\label{eq:2Lemma3}
\sum_{k = 0}^{s_n}  \binom{s_n}{k} \left| p^k (1 - p)^{s_n - k} - q^k (1 - q)^{s_n - k} \right|=\\
  \sum_{k = k_0 + 1}^{s_n} \binom{s_n}{k} \left( p^k (1 - p)^{s_n - k} - q^k (1 - q)^{s_n - k} \right) - \sum_{k = 0}^{k_0} \binom{s_n}{k} \left( p^k (1 - p)^{s_n - k} - q^k (1 - q)^{s_n - k} \right). 
\end{multline}
Consider the function 
$$
f(x) = \sum\limits_{k = k_0 + 1}^{s_n} \binom{s_n}{k} x^k (1 - x)^{s_n - k} - \sum\limits_{k = 0}^{k_0} \binom{s_n}{k} x^k (1 - x)^{s_n - k}.
$$
By the Lagrange's mean value theorem, we have that there is a number $t \in (q, p)$ such that 
\begin{align*}
\sum_{k = k_0 + 1}^{s_n} \binom{s_n}{k} &  \left( p^k (1 - p)^{s_n - k} - q^k (1 - q)^{s_n - k} \right) 
 - \sum_{k = 0}^{k_0} \binom{s_n}{k} \left( p^k (1 - p)^{s_n - k} - q^k (1 - q)^{s_n - k} \right) \\
 &= f(p) - f(q) = (p - q) f'(t) \\
 &=(p - q) \sum_{k = k_0 + 1}^{s_n} \binom{s_n}{k} \left( k t^{k - 1} (1 - t)^{s_n - k} - (s_n - k) t^k (1 - t)^{s_n - k - 1} \right) - \\
&\quad\quad\quad\quad - (p - q) \sum_{k = 0}^{k_0} \binom{s_n}{k} \left( k t^{k - 1} (1 - t)^{s_n - k} - (s_n - k) t^k (1 - t)^{s_n - k - 1} \right).
\end{align*}
The last expression equals
\begin{align}
\label{eq:3Lemma3}
 (p -  q)  s_n  \sum_{k = k_0}^{s_n - 1} \binom{s_n - 1}{k} t^k (1 - t)^{s_n - 1 - k} - (p - q) s_n \sum_{k = k_0 + 1}^{s_n - 1} \binom{s_n - 1}{k} t^k (1 - t)^{s_n - 1 - k} - \notag \\
 - (p - q) s_n \sum_{k = 0}^{k_0 - 1} \binom{s_n - 1}{k} t^k (1 - t)^{s_n - 1 - k} + (p - q) s_n \sum_{k = 0}^{k_0} \binom{s_n - 1}{k} t^k (1 - t)^{s_n - 1 - k} \\ 
 = 2 (p - q) s_n \binom{s_n - 1}{k_0} t^{k_0} (1 - t)^{s_n - 1 - k_0}. \notag
\end{align}
By Stirling's formula, there exists a constant $C'>0$ such that, for all non-negative integers $a>b$ and any real $x\in(0,1)$,
\begin{align*}
 {a\choose b}x^b(1-x)^{a-b}  \leq C' \sqrt{\frac{a}{2 \pi (a - b) b}} \left(\frac{ax}{b}\right)^{b} \left(\frac{a(1 - x)}{a - b}\right)^{a - b} 
   \leq C'\sqrt{\frac{a}{2\pi(a-b)b}}
\end{align*}
since the function $x^b(1-x)^{a-b}$ achieves its maximum at $x=\frac{b}{a}$.

For $p_{\min} \geq \frac{3}{2s_n}$, we have $k_0 < ps_n \leq (1 - p_{\min})s_n < s_n - 1$. Hence, for $k_0 < s_n - 1$
\begin{align}
\label{eq:4Lemma3}
2 (p - q) s_n  \binom{s_n - 1}{k_0} t^{k_0} (1 - t)^{s_n - 1 - k_0} 
& \leq 2 (p - q) s_n C' \sqrt{\frac{s_n - 1}{2 \pi (s_n - 1 - k_0) k_0}} \notag  \\
& = 2 (p - q) C' \sqrt{\frac{s_n - 1}{2 \pi (1 - \frac{k_0 + 1}{s_n}) \frac{k_0}{s_n}}}.
\end{align}
The inequality $\left(1 - \frac{k_0 + 1}{s_n}\right) \frac{k_0}{s_n} \geq \frac{1}{2}\left(1 - \frac{1}{s_n}\right)\min\left\{\left(1 - \frac{k_0 + 1}{s_n}\right), \frac{k_0}{s_n}\right\}$ holds because the left side is a product of two factors that sum up to $1 - \frac{1}{s_n}$, and then the largest one is at least half of this sum. Due to~\eqref{eq:k_0Lemma3}, $p s_n > k_0 \geq \lfloor q s_n \rfloor > q s_n - 1$. Therefore,
$$
\min\left\{\left(1 - \frac{k_0 + 1}{s_n}\right), \frac{k_0}{s_n}\right\} \geq p_{\min} - \frac{1}{s_n}.
$$
Finally, from \eqref{eq:1Lemma3}, \eqref{eq:2Lemma3}, \eqref{eq:3Lemma3}, \eqref{eq:4Lemma3}, we get
$$
|\mathrm{Pr}((\xi_1,\ldots,\xi_{s_n}) \in A_n) - \mathrm{Pr}((\eta_1,\ldots,\eta_{s_n}) \in A_n)|
 \leq 2 C' \sqrt{\frac{s_n}{\pi (p_{\min} - \frac{1}{s_n})}} |p - q|.
$$

Since $p_{\min} - \frac{1}{s_n} \geq \frac{1}{3} p_{\min}$, we have the inequality 
\begin{equation}
|\mathrm{Pr}((\xi_1,\ldots,\xi_{s_n}) \in A_n) - \mathrm{Pr}((\eta_1,\ldots,\eta_{s_n}) \in A_n)|
 \leq 2 C' \sqrt{\frac{3 s_n}{\pi p_{\min}}} |p - q|. \label{eq:5Lemma3}
\end{equation}
If $p_{\min} < \frac{3}{2s_n}$, note that 
\begin{align}
|\mathrm{Pr}((\xi_1,\ldots,\xi_{s_n}) \in A_n) - \mathrm{Pr}((\eta_1,\ldots,\eta_{s_n}) \in A_n)| \leq 1. \label{eq:6Lemma3}
\end{align}
Hence, for $|p - q| \geq \frac{1}{s_n}$, we have
$$
|\mathrm{Pr}((\xi_1,\ldots,\xi_{s_n}) \in A_n) - \mathrm{Pr}((\eta_1,\ldots,\eta_{s_n}) \in A_n)|
\leq \sqrt{\frac{3 s_n}{2 p_{\min}}} |p - q|.
$$
For $|p - q| < \frac{1}{s_n}$, we can claim that $p, q \leq \frac{3}{s_n}$. Therefore, 
\begin{align}
\label{eq:7Lemma3}
 &\sum_{k = 0}^{s_n}  \binom{s_n}{k} \left| p^k (1 - p)^{s_n - k} - q^k (1 - q)^{s_n - k} \right| \notag \\
   &\leq \sum_{k = 0}^{s_n} \binom{s_n}{k} \left| p^k (1 - p)^{s_n - k} - q^k (1 - p)^{s_n - k} \right| + \sum_{k = 0}^{s_n} \binom{s_n}{k} \left| q^k (1 - p)^{s_n - k} - q^k (1 - q)^{s_n - k} \right| \notag \\
& \leq \sum_{k = 0}^{s_n} \binom{s_n}{k} | p^k - q^k | + \sum_{k = 0}^{s_n} \binom{s_n}{k} s_n q^k |p - q| \notag \\
& \leq \sum_{k = 0}^{s_n} \binom{s_n}{k} k |p - q| \left(\frac{3}{s_n}\right)^{k - 1} + \sum_{k = 0}^{s_n} 3 \binom{s_n}{k} \left(\frac{3}{s_n}\right)^{k-1} |p - q|\\
& = \sum_{k = 0}^{s_n} \binom{s_n}{k} |p - q| \left(\frac{3}{s_n}\right)^{k - 1} (k + 3) \notag \\
& \leq s_n |p - q| \sum_{k = 0}^{s_n} \frac{3^{k - 1}}{k!} (k + 3) \leq 2 e^3 s_n |p - q|  < 2 e^3 \sqrt{\frac{3s_n}{2 p_{\min}}} |p - q|. \notag
\end{align}
Let $C'' = \max\left\{ 2 C' \sqrt{\frac{3}{\pi}}, \sqrt{\frac{3}{2}}, 2 e^3 \sqrt{\frac{3}{2}}\right\}$. Then, from \eqref{eq:1Lemma3}, \eqref{eq:5Lemma3}, \eqref{eq:6Lemma3}, and \eqref{eq:7Lemma3}, we get
$$
|\mathrm{Pr}((\xi_1,\ldots,\xi_{s_n}) \in A_n) - \mathrm{Pr}((\eta_1,\ldots,\eta_{s_n}) \in A_n)| \leq C'' \sqrt{\frac{s_n}{p_{\min}}} |p - q|,
$$
completing the proof of the lemma.

\end{document}